\documentclass[11pt,letterpaper]{article}

\usepackage[utf8]{inputenc}
\usepackage{lmodern} %
\usepackage{authblk}
\usepackage[bookmarks,colorlinks,breaklinks]{hyperref}
\hypersetup{urlcolor=blue, colorlinks=true, citecolor=green!50!black, linkcolor=blue}
\usepackage{multirow}
\usepackage[T1]{fontenc} %
\usepackage{amsmath}
\usepackage{amsfonts}
\usepackage{amssymb}
\usepackage{amsthm}
\usepackage{thmtools}
\usepackage{thm-restate}
\usepackage{pgffor}
\usepackage{xcolor}
\usepackage[lined,boxed,ruled,norelsize,algo2e,linesnumbered,noend]{algorithm2e}
\usepackage{cleveref}
\usepackage{graphicx}
\usepackage{geometry}
\usepackage{enumitem}
\usepackage{url} %
\usepackage{tikz}

\ifdefined\DEBUG
	\usepackage{showlabels}
	\def\jan#1{\textcolor{red}{JAN: #1}}
    \def\daniel#1{\textcolor{blue}{Daniel: #1}}
\else
    \def\jan#1{}
    \def\daniel#1{}
\fi

\declaretheorem[numberwithin=section,refname={Theorem,Theorems},Refname={Theorem,Theorems}]{theorem}
\declaretheorem[numberlike=theorem]{lemma}

\declaretheorem[numberlike=theorem]{corollary}
\declaretheorem[numberlike=theorem]{definition}

\declaretheorem[numberlike=theorem,Refname={Fact,Facts}]{fact}

\renewcommand{\th}{^{th}}
\newcommand{\tO}{\tilde{O}}
\newcommand{\hO}{\widehat{O}}

\newcommand{\R}{\mathbb{R}}

\renewcommand{\tilde}{\widetilde}
\renewcommand{\hat}{\widehat}

\newcommand{\poly}{\operatorname{poly}}
\newcommand{\polylog}{\operatorname{polylog}}
\newcommand{\nnz}{\operatorname{nnz}}%

\foreach \x in {A,...,Z}{%
	\expandafter\xdef\csname m\x\endcsname{\noexpand\mathbf{\x}}
}

\foreach \x in {A,...,Z}{%
	\expandafter\xdef\csname om\x\endcsname{\noexpand\overline{\noexpand\mathbf{\x}}}
}

\foreach \x in {A,...,Z}{%
	\expandafter\xdef\csname c\x\endcsname{\noexpand\mathcal{\x}}
}

\newcommand{\otau}{\noexpand{\overline{\tau}}}
\newcommand{\omu}{\noexpand{\overline{\mu}}}
\newcommand{\osigma}{\noexpand{\overline{\sigma}}}%
\newcommand{\os}{\noexpand{\overline{s}}}
\newcommand{\ot}{\noexpand{\overline{t}}}
\newcommand{\og}{\noexpand{\overline{g}}}
\newcommand{\ou}{\noexpand{\overline{u}}}
\newcommand{\ov}{\noexpand{\overline{v}}}
\newcommand{\ow}{\noexpand{\overline{w}}}
\newcommand{\ox}{\noexpand{\overline{x}}}
\newcommand{\oy}{\noexpand{\overline{y}}}
\newcommand{\oz}{\noexpand{\overline{z}}}
\newcommand{\oc}{\noexpand{\overline{c}}}

\newcommand{\init}{\mathrm{(init)}}%
\newcommand{\target}{\mathrm{(end)}}%
\newcommand{\final}{\mathrm{(final)}}%

\makeatletter
\renewcommand{\paragraph}{%
	\@startsection{paragraph}{4}%
	{\z@}{1.25ex \@plus 1ex \@minus .2ex}{-1em}%
	{\normalfont\normalsize\bfseries}%
}
\makeatother

\title{Faster High Accuracy Multi-Commodity Flow\\from Single-Commodity Techniques}
\author[1]{Jan van den Brand}
\affil[1]{Georgia Institute of Technology}
\author[1]{Daniel Zhang}

\begin{document}
\pagenumbering{roman}
\maketitle
\begin{abstract}
    
Since the development of efficient linear program solvers in the 80s, all major improvements for solving multi-commodity flows to high accuracy came from improvements to general linear program solvers.
This differs from the single commodity problem (e.g.~maximum flow) where all recent improvements also rely on graph specific techniques such as graph decompositions or the Laplacian paradigm (see e.g.~\cite{CohenMSV17,KathuriaLS20,BrandLL+21,ChenKL+22}).

This phenomenon sparked research to understand why these graph techniques are unlikely to help for multi-commodity flow. 
\cite{KyngZ20} reduced solving multi-commodity Laplacians to general linear systems and \cite{DingKZ22} showed that general linear programs can be reduced to 2-commodity flow.
However, the reductions create sparse graph instances, so improvement to multi-commodity flows on denser graphs might exist. %

We show that one can indeed speed up multi-commodity flow algorithms on non-sparse graphs using graph techniques from single-commodity flow algorithms.
This is the first improvement to high accuracy multi-commodity flow algorithms that does not just stem from improvements to general linear program solvers. %
In particular, using graph data structures from recent min-cost flow algorithm by \cite{BrandLL+21} based on the celebrated expander decomposition framework, we show that 2-commodity flow on an $n$-vertex $m$-edge graph can be solved in $\tilde{O}(\sqrt{m}n^{\omega-1/2})$ time for current bounds on fast matrix multiplication $\omega \approx 2.373$, improving upon the previous fastest algorithms with $\tilde{O}(m^\omega)$ \cite{CohenLS19} and $\tilde{O}(\sqrt{m}n^2)$ \cite{KapoorV96} time complexity.
For general $k$ commodities, our algorithm runs in $\tilde{O}(k^{2.5}\sqrt{m}n^{\omega-1/2})$ time.

\end{abstract}
\newpage
\tableofcontents
\newpage
\pagenumbering{arabic}

\section{Introduction}
The multi-commodity flow problem arises when more than one commodity must be routed in a shared network.
These types of problems occur often in traffic-, distribution-, and communication-systems, as well as in network and VLSI design. %
Formally, we are given a (possibly directed) graph $G=(V,E)$ with edge capacities $u\in\R_{>0}^E$ and multiple source-sink pairs (one for each commodity) and the task is to route a flow for each source-sink pair, such that the sum of the flows do not exceed the edge capacities.
In the maximum through-put variant, the task is to maximize the sum of flows routed between the pairs, whereas in the minimum cost variant, each source-sink pair has a demand that must be satisfied and has costs assigned to the edges. The task here is to route the flow such that it minimizes the cost.

There exist various algorithms that solve these problems to $(1+\epsilon)$ accuracy 
\cite{%
LeightonMPSST95,%
Fleischer00,%
GargK07,%
Madry10,%
KelnerMP12,%
Sherman13,%
KelnerLOS14,%
Peng16}.
On undirected graphs, the problem can even be solved in nearly-linear time \cite{%
Sherman13,%
KelnerLOS14,%
Peng16,
ChenY23}.
However, all these algorithms are either low-accuracy (i.e.~their complexity has a $\poly(1/\epsilon)$ dependence), or they consider relaxations such as flows without capacities minimizing a mixed $\ell_{q,p}$-norm for constant $p,q$ \cite{ChenKL+22}.
To obtain high-accuracy solutions (with $\polylog(1/\epsilon)$ time complexity dependence), the multi-commodity flow problems can be phrased as a linear program (LP) and thus solved via efficient linear programming solvers. %
For instance, on sparse graphs, the current fastest algorithm for multi-commodity flow are the recent matrix multiplication time linear program solvers \cite{CohenLS19,LeeSZ19,Brand20,JiangSWZ21}, running in $\tilde{O}((km)^\omega)$ time\footnote{%
The matrix exponent $m^\omega$ is the number of operations required to multiply two $m\times m$ matrices. Currently $\omega\le2.373$ \cite{AlmanW21}. We write $\omega(\cdot,\cdot,\cdot)$ for the rectangular matrix multiplication exponent \cite{GallU18}. Here $n^{\omega(a,b,c)}$ is the complexity for multiplying an $n^a\times n^b$ matrix by $n^b \times n^c$ matrix.}
on $m$-edge graphs with $k$ commodities.
Before the advent of polynomial time LP solvers, 
Hu developed a polynomial time algorithm for multi-commodity flow on undirected graphs \cite{Hu63}. 
But since linear programs became efficient to solve, all improvements for solving multi-commodity flow to high accuracy stem from improvements to LP solvers
\cite{Karmarkar84,Renegar88,Vaidya87,Vaidya89,KapoorV96,LeeS14,LeeS15,CohenLS19,LeeSZ19,Brand20,JiangSWZ21}. 
In particular, no improvements were made via graph specific techniques.

This differs a lot from the developments that occurred for single-commodity flows (i.e.~the classical maximum flow and min-cost flow problems).
While recent developments rely on continuous optimization methods similar to LP solvers, they do use many powerful graph specific techniques \cite{DaitchS08,Madry13,CohenMSV17,Madry16,AxiotisMV21,LeeS14,LiuS20,KathuriaLS20,BrandLN+20,BrandLL+21,DongGGLPSY22,BrandGJLLPS22,GaoLP21,ChenKL+22}.
For example, while the central path method reduces solving general LPs to solving a sequence of linear systems,
applying this method to single-commodity flows results in a sequence of Laplacian systems.
Such a system can be solved in near-linear time using Laplacian solvers \cite{SpielmanT04}.
This powerful framework of combining continuous optimization methods with Laplacian solvers is often referred to as ``Laplacian paradigm''.
Additionally, even the central path method itself can be accelerated for the special case when the LP is a single-commodity flow instance \cite{Madry13,Madry16,CohenMSV17,LiuS20,AxiotisMV21,KathuriaLS20}.
Every recent improvement for single-commodity flows stems from combining continuous optimization techniques with graph specific tools.
Yet somehow, none of these tools translate from single-commodity to multi-commodity problems. Even when extending to just two commodities, these tools have not found any application in the high accuracy regime.
This has raised the following question:
\begin{center}
    \emph{Can single-commodity techniques be used to improve high-accuracy multi-commodity flow algorithms?}
\end{center}
This sparked research to better understand why so far these graph techniques could not help for the high accuracy multi-commodity flow problem. 
Kyng and Zhang showed that an equivalent of the Laplacian paradigm for multi-commodity flows is unlikely \cite{KyngZ20}. 
They showed that any linear system can be reduced to a 2-commodity Laplacian. 
So despite near-linear time solvers for Laplacian systems, such solvers are unlikely to exist for 2-commodity Laplacians unless we can solve all linear systems in near linear time. 
Thus it is unlikely that we can speed up the central path method in a similar way as the Laplacian paradigm did for single-commodity flow.
While this is only an argument against one specific algorithmic approach, there is also evidence that 2-commodity flows in general are hard.
Itai showed that any linear program can be reduced to exact 2-commodity flow \cite{Itai78}, and recently Ding, Kyng, Zhang \cite{DingKZ22} extended this result to the high-accuracy regime.
Thus any algorithm that solves 2-commodity flow to high accuracy can also be used to solve general linear programs to high accuracy.
The reduction by \cite{DingKZ22} produces a 2-commodity flow instance on a sparse graph
where the number of edges corresponds to the number of non-zeros in the LP. %
So, if one can get an algorithm using graph techniques that is competitive with general purpose linear program solvers on the 2-commodity flow problem on sparse graphs, 
then this algorithm would also be competitive on general sparse LPs, despite the LP not having any kind of graph structure.
With these insights from \cite{Itai78,DingKZ22,KyngZ20}, it seems that improvements to multi-commodity flow via graph techniques are impossible. However, since these reductions produce sparse graphs, single-commodity techniques might still lead to improvements on dense graphs.

In this work we show that this is indeed possible.
We obtain the first
improvement to the multi-commodity flow problem via graph specific single-commodity techniques since the development of efficient LP solvers in the 80s.
This is the first improvement that does not just stem from improvements to general linear program solvers.

Our algorithm combines algebraic methods from general purpose linear program solvers \cite{CohenLS19,Brand20,JiangSWZ21} with the dynamic expander decomposition framework from dynamic graph theory \cite{HuaKGW22,BernsteinBGNSS022,NanongkaiS17,WulffNilsen17,NanongkaiSW17,SaranurakW19}.
This graph technique was also used in all recent improvements to single-commodity flows (either directly or inside a Laplacian solver) \cite{DaitchS08,Madry13,CohenMSV17,Madry16,AxiotisMV21,LeeS14,LiuS20,KathuriaLS20,BrandLN+20,BrandLL+21,DongGGLPSY22,BrandGJLLPS22,GaoLP21,ChenKL+22}.

\subsection{High Accuracy Results}

As mentioned before, multi-commodity flows can be solved to high-accuracy using linear program solvers. 
For an $n$-node $m$-edge graph with $k$ commodities, the respective linear program has $km$ variables $kn+m$ equality-constraints. %
Using state-of-the-art linear program solvers, the following complexities can be achieved for solving multi-commodity flow with polynomially bounded edge capacities $U\le\poly(n)$ and polynomially bounded error $\epsilon^{-1} \le \poly(n)$: 
$\tilde O((km)^\omega)$ time \cite{CohenLS19},
$\tilde O(k^{2.5}\sqrt{m}n^2)$ time \cite{KapoorV96},
$\tilde O((kn+m)^{2.5})$ time \cite{LeeS14,LeeS15}\footnote{For single-commodity flows, Lee and Sidford's algorithm \cite{LeeS14} runs in $\sqrt{n}$ iterations. However, for multi-commodity flow the number of iterations of \cite{LeeS14} is $\sqrt{m+kn} > \sqrt{m}$, see e.g.~\cite[Section 7.4]{LeeS15}.}.

Using expander decomposition, we improve these complexities on graphs that are at least slightly dense. This is the first improvement to multi-commodity flow that stems from single-commodity techniques. 
All previous improvements \cite{Karmarkar84,Renegar88,Vaidya87,Vaidya89,KapoorV96,LeeS14,LeeS15,CohenLS19} stem from improvements to general LP solvers.

A conceptual insight of our work is that graph techniques can lead to faster multi-commodity flow algorithms despite increasing evidence to the contrary.
Our improvements are possible, because previous impossibility results hold only for sparse graphs.
We show that dense multi-commodity flow is more combinatorial in nature than previously thought.

\begin{restatable}{theorem}{throughput}\label{thm:intro:throughput}
For any $0\le\mu\le 1$, 
given a $k$-commodity instance on graph $G=(v,E)$ with integer edge capacities $u\in[0,U]^E$ and source sink pairs $(s_1,t_1),...,(s_k,t_k)\in V\times V$, we can solve maximum through-put commodity flow \emph{deterministically} up to additive error $\epsilon>0$ in time
$$
\tilde{O}(k^{2.5}\sqrt{m}(n^{\omega-1/2}+n^{\omega(1,1,\mu)-\mu/2}+n^{1+\mu}+n\log( U/\epsilon)) \log\frac{U}{\epsilon}).
$$
For current bounds on $\omega(\cdot,\cdot,\cdot)$ \cite{AlmanW21,GallU18}, and polynomially bounded $u,\epsilon^{-1}$, this is $$
\tilde{O}(k^{2.5}\sqrt{m}n^{\omega-1/2}).
$$
\end{restatable}

\cite{Itai78,KyngZ20,DingKZ22} show that hardness of multi-commodity flow is already given by just two commodities, so let us for now focus on $k=2$.
Our new algorithm improves upon previous work whenever the graph is at least slightly dense. %
On sparse graphs, we match the $\tilde{O}(m^\omega)$-time\footnote{Like our results, these algorithms run in $\tilde{O}(m^\omega)$ time for current bounds on matrix multiplication. In particular, \cite{CohenLS19,LeeSZ19,Brand20} have the same $\tilde{O}(n^{\omega-1/2}+n^{\omega(1,1,\mu)-\mu/2}+n^{1+\mu})$ dependency as our \Cref{thm:intro:throughput}.} algorithm by \cite{CohenLS19,LeeSZ19,Brand20,JiangSWZ21} up to %
the extra $\sqrt{m}n\log( U/\epsilon)$ term, though this term does not matter for polynomially bounded $U$ and $1/\epsilon$.
On dense graphs, we improve upon the $\tilde{O}(\sqrt{m}n^2 \log U/\epsilon)$-time algorithm by \cite{KapoorV96}.

We remark that the reduction of Ding, Kyng, Zhang \cite{DingKZ22} creates $2$-commodity flow instances on sparse graphs. So our algorithm does not improve upon general purpose LP solvers, but using their reduction, our algorithm can solve LPs within the same $\tilde{O}(m^\omega \log \delta^{-1})$-time complexity\footnote{These LP solvers return an approximate solution with additive $\delta$-error.} as other state-of-the-art LP solvers \cite{CohenLS19,LeeSZ19,Brand20,JiangSWZ21} on sparse LPs.
In particular, if we were able to improve upon \cite{CohenLS19,LeeSZ19,Brand20,JiangSWZ21} for $2$-commodity flow on sparse graphs, then our algorithm would improve general LP solvers as well.
This gives strong evidence that improvements via single-commodity techniques might only be possible for at least slightly dense graphs.

Ding, Kyng and Zhang \cite{DingKZ22} posed the open problem whether their reduction from LP
to 2-commodity flow could be modified to maintain the ``shape'' of the LP; 
If a tall sparse LP is given ($n$ rows, $m$ columns, $n \le m \le n^2$)
can one construct a multi-commodity flow instance on $\tilde{O}(m)$ edges and $\tilde{O}(n)$ vertices? %
If such a reduction exists, then our multi-commodity flow algorithm would beat state-of-the-art LP solvers when $m\le n^{1.254}$.
So either there is exciting opportunity to improve general LP solvers using graph theoretic expander decomposition techniques, or our result can be interpreted as evidence that such a reduction is not possible.

Our result can also solve the minimum cost variant of $k$-commodity flow.
\begin{restatable}{theorem}{mincost}\label{thm:intro:mincost}
For any $0\le\mu\le 1$, given a $k$-commodity instance on graph $G=(V,E)$ with integer edge capacities $u\in[0,U]^E$, integer costs $c_1,...,c_k\in[-C,C]^E$ and integer demands $d_1,...,d_k\in[-U,U]^V$, 
we can solve minimum-cost commodity flow \emph{deterministically} up to additive error $\epsilon>0$ in time
$$
\tilde{O}(k^{2.5}\sqrt{m}(n^{\omega-1/2}+n^{\omega(1,1,\mu)-\mu/2}+n^{1+\mu}+n\log( CU/\epsilon)) \log\frac{CU}{\epsilon}).
$$
For current bounds on $\omega$ \cite{AlmanW21,GallU18} and polynomially bounded $\epsilon^{-1},C,U$, this is $$
\tilde{O}(k^{2.5}\sqrt{m}n^{\omega-1/2}).
$$
The returned flows $f_1,...,f_k\in\R^E_{\ge0}$ satisfy the demands approximately with (here $\mB\in\R^{E\times V}$ is the incidence matrix)
$$
\|\mB^\top f_i -d_i \|_1 \le \epsilon \text{ for }i=1,...,k.
$$
\end{restatable}

We remark that the small error w.r.t.~the demands also occurs when solving multi-commodity flow with high accuracy LP solvers such as \cite{CohenLS19,LeeSZ19,Brand20,JiangSWZ21}.

\subsection{Techniques} %
\label{sec:intro:techniques}

Here we summarize the techniques used by our multi-commodity flow algorithm and how they relate to previous general purpose LP solvers. 
A more detailed outline of the techniques is given in \Cref{sec:overview}.
Our algorithm is based on general purpose LP solvers which we accelerate for the multi-commodity flow problem using single-commodity flow techniques.
To highlight our ideas we start with a quick summary on how general purpose LP solvers were accelerated for single-commodity flow and then explain how a similar improvement can be made for multi-commodity flow.

In the line of work of \cite{BrandLSS20,BrandLL+21,BrandLN+20}, it was shown that a linear program with constraint matrix of dimension $m\times n$ (with $m\ge n$) can be solved in $\tilde{O}(mn+n^{2.5})$ time.\footnote{For simplicity we hide $\log U/\epsilon$ terms (for single and multi-commodity flow) and $\log\delta^{-1}$ terms (for general LP solvers) in this overview.\label{footnote}}
This algorithm uses data structures that work for general matrices. 
If the linear program is a single-commodity flow problem, then the additional graph structure of the involved matrices allows for faster data structures using graph techniques. 
This resulted in an $\tilde{O}(m+n^{1.5})$ time algorithm (i.e.~saved a factor $n$) for single-commodity flow \cite{BrandLL+21}.
A natural question is if a similar speed up is possible when extending from single to multi-commodity flow.
For this, we will focus on each term ($mn$ and $n^{2.5}$) separately and discuss how they can be improved for multi-commodity flows.

The two terms $mn$ and $n^{2.5}$ generally come from the following two factors:
(i) How many iterations does the algorithm take, and the time required to solve a linear system in each iteration (this is where the $\tilde{O}(n^{2.5})$ term came from),
(ii) How much time is required for a certain ``heavy hitter problem'' which we describe next. (This is where the $\tilde{O}(mn)$ term came from)

\paragraph{Heavy Hitter}
LP solvers based on the central path method must %
compute a matrix vector product of the form $\mA h$ in each iteration.
Here $\mA$ is a fixed matrix %
and $h$ is some new vector given in each iteration. 
It was shown in \cite{BrandLSS20,BrandLN+20} that it suffices to just detect large entries of this product, a task they refer to as ``heavy hitters''.

In \cite{BrandLSS20}, this heavy hitter task requires $\tilde{O}(mn)$ total time over all iterations and \cite{BrandLN+20} shows that for the special case where $\mA$ is an edge-vertex-incidence matrix, it can be solved in just $\tilde{O}(m)$ total time over all iterations.

For our application of $k$-commodity flow, this matrix $\mA$ will be of the following form: 
The rows of $\mA$ are given by Kronecker-products. For each edge $(u,v)\in E$ we have $k+1$ rows in $\mA$ given by $\mW^{(u,v)} \otimes (e_u^\top - e_v^\top)$ for some matrix $\mW^{(u,v)}\in\R^{(k+1)\times(k+1)}_{>0}$ and $e_u,e_v \in \R^{n}$ standard unit vectors. 
So matrix $\mA$ can be seen as an incidence matrix where instead of real edge weights, we now have $(k+1)^2$-dimensional edge weights $\mW^{(u,v)}$. 
This weight matrix $\mW^{(u,v)}$ has a very specific structure given by the following definition. 
We show that the heavy hitter task is efficiently solvable on matrices of the following form.

\begin{definition}\label{def:intro:mcincidence}
Given graph $G=(V,E)$ and vectors $h^{(u,v)} \in \R^{k+1}_{>0}$ for each $(u,v) \in E$.
We call the matrix $\mM\in\R^{(k+1)m\times(k+1)n}$ with rows given by 
$$\left(\mH^{(u,v)})^{1/2}\left(\mI - \frac{\mathbf{1}_{k+1} h^\top}{\|h\|_1}\right)\right) \otimes (e_u^\top - e_v^\top) \in \R^{k+1\times (k+1)n}$$
a $k$-commodity incidence matrix.\footnote{%
The notion of a multi-commodity incidence matrix was defined in a different way in \cite{KyngZ20}, however $\mM^\top\mM$ is a multi-commodity Laplacian as defined in \cite{KyngZ20}. %
} 

Here $\mH^{(u,v)}$ is diagonal matrix with entries of $h^{(u,v)}$ on the diagonal, $\mI$ is identity matrix, $\mathbf{1}_{k+1}$ is the $k+1$ dimensional all 1-vector, and $\otimes$ is the Kronecker-product.
\end{definition}
Unlike heavy hitters from previous work \cite{BrandLSS20,BrandLL+21,BrandLN+20} where the matrix was fixed (though scaling of rows was admissible), the matrix from \Cref{def:intro:mcincidence} is allowed to change over time because we allow to update entries of any $d^{(u,v)}$.

Additionally, note that the matrix from \Cref{def:intro:mcincidence} has $O(k)$ non-zero entries per row.
Solving the heavy hitter problem on general sparse matrices is one of the major open problems in linear program solving. %
In particular, if heavy hitter could be solved in $O(km)$ total time on a matrix with $m$ rows and $k$ non-zero entries per row for $1\le k \le n$, one would obtain a nearly linear time solver for sparse and sufficiently tall linear programs \cite{BrandLL+21}. So far there are only nearly linear time algorithms for dense linear programs \cite{BrandLSS20,BrandLL+21}.
Our heavy hitter for \Cref{def:intro:mcincidence} is the first result for a structured special case of sparse matrices with $>2$ non-zero entries per row. Though with the additional structure imposed on $\mW^{(u,v)} = \left((\mH^{(u,v)})^{1/2}\left(\mI - \frac{\mathbf{1}_{k+1} h^\top}{\|d\|_1}\right)\right)$ it does not improve general linear program solvers.
It is interesting to see though, that the heavy hitter problem for multi-commodity incidence matrices seems to be easier than the heavy hitter problem for general sparse matrices.
For comparison: solving multi-commodity Laplacian systems (such as $\mM^\top \mM x = b$ for $\mM$ from \Cref{def:intro:mcincidence}) are as hard as general linear systems \cite{KyngZ20}. And solving linear programs are as hard as sparse multi-commodity flows \cite{Itai78,DingKZ22}. Yet, heavy hitter for multi-commodity incidence matrices (\Cref{def:intro:mcincidence}) do not seem to be as hard as general sparse heavy hitters.
In particular, using the special structure of $\mW^{(u,v)}$, we can reduce the heavy hitter task on multi-commodity incidence matrices to the heavy hitter task on classical incidence matrices.

The hardness of solving linear systems in multi-commodity Laplacians such as $\mM^\top \mM$ brings us to the next section.

\paragraph{Solving Linear Systems} %
In addition of detecting the heavy hitters of some matrix vector product $\mA h$, modern linear program solvers must also solve a linear system in each iteration.

We already stated that tall LPs of size $m\times n$ ($m \gg n$) could be solved in $\tilde{O}(mn+n^{2.5})$ time.
Here the $n^{2.5}$ term comes from the LP solver internally running  $\tilde{O}(\sqrt{n})$ iterations and in each iteration solving a linear system (i.e~$O(n^2)$ time to multiply a vector with some $n\times n$ sized matrix inverse).
In the case of single-commodity flow, this linear system is given by a Laplacian matrix and by using Laplacian solvers the system the $O(n^2)$ cost can be reduced to just $\tilde O(n)$ time.

A natural idea is to use same techniques to speed up multi-commodity flow. 
However, there are two problems for multi commodity flow, that result in a slow down.
First, the number of iterations is larger: 
while \cite{LeeS14,BrandLL+21,BrandLN+20,LeeS14} consider tall LPs of size $m\times n$, multi-commodity LPs have dimension $(k+1)m\times kn+m$, so the number of iterations is $\sqrt{kn+m}$.
Especially for small $k$ there is no substantial speed up from using \cite{LeeS14} compared to, let's say, the classic log-barrier method with $\sqrt{km}$ iterations.
For simplicity we will use the simpler but slightly slower $\sqrt{km}$ iteration algorithm.

The second and bigger issue is that for multi-commodity flow, we cannot use a Laplacian solver to solve the linear system. \cite{KyngZ20} shows that for multi-commodity flow, solving this linear system 
is as hard as solving a general linear system, even for just $k=2$ commodities. 
So it is unlikely that there is a fast solver using graph techniques similar to Laplacian solvers.

With these two issues, one would achieve only $\tilde{O}(\sqrt{m}n^2)$ time for $2$-commodity flow ($\sqrt{m}$ iterations and $n^2$ for solving the linear system), as was already achieved in \cite{KapoorV96}.

Our algorithm also cannot use graph techniques to speed up solving this linear system. 
However, we show that because of the sparsity and structure of the linear program, we can solve the linear system in $\tilde O(k^2 n^{\omega-1/2})$ amortized time per iteration instead of $O(k^2n^2)$.
This uses projection maintenance techniques from \cite{CohenLS19,LeeSZ19,Brand20,Brand21,JiangSWZ21} which, when applied directly to the linear system, would solve it 
in $\tilde{O}((km)^{\omega-1/2})$ time per iteration.
The improvement from $\tilde{O}((km)^{\omega-1/2})$ to $\tilde{O}(k^2n^{\omega-1/2})$ comes from two observations:
(i) The sparsity of the LP implies that any change to the linear system from one iteration to the next is very sparse.
(ii) Most constraints of the linear program are of form $\sum_{i=1}^k f_i\le u$ (i.e.~the capacity constraints of $k$-commodity flow). 
These constraints induce a certain structure in the linear system that can be exploited to reduce the size onto a smaller $kn\times kn$ sized linear system for $k$-commodity flow. 
In particular, one can transform the linear system to be of shape $\mM^\top\mM x = b$ for matrix $\mM$ being a multi-commodity incidence matrix as in \Cref{def:intro:mcincidence}.
Property (ii) was previously used by \cite{KapoorV96} to accelerate Vaidya's and Karmarkar's LP solvers \cite{Vaidya87,Karmarkar84} for $k$-commodity flow, which is why they achieved an $\tilde{O}(k^{2.5}\sqrt{m}n^2)$ time algorithm. 
We now use this idea together with (i) to accelerate the projection maintenance/inverse maintenance techniques from \cite{CohenLS19,LeeSZ19,Brand20,Brand21,JiangSWZ21}.
The next difficulty for obtaining a faster algorithm comes from projecting the smaller $kn$ dimensional solution back onto $O(km)$ dimensional space without paying $O(km)$ time per iteration (and thus leading to an at best $\tilde{O}((km)^{1.5}) \le \tilde{O}(k^{1.5}n^3)$ time algorithm, since we have $\tilde{O}(\sqrt{km})$ iterations).
We show that projecting the solution back onto $O(km)$ dimensional space is solved by the ``heavy hitter'' problem on multi-commodity incidence matrices.

\subsection{Related work}\label{sec:related}
While we focus on solving multi-commodity flow to high accuracy, related work also studies low accuracy regimes \cite{%
LeightonMPSST95,%
Fleischer00,%
GargK07,%
Madry10,%
KelnerMP12,%
Sherman13,%
KelnerLOS14,%
Peng16}. 
On undirected graphs, the ``maximum concurrent flow'' (i.e.~maximizing $F$ s.t.~each commodity has at least $F$ units of flow) for $k$ commodities can be reduced to $2^{k-1}$ single commodity flows \cite{RothschildW66a,ChenY23}. 
Maximum concurrent flow can be solved by our algorithms by adding edges $t_i\to s_i$ of capacity $F$ and negative cost for each of the $k$ source/sink pairs, then binary searching for the maximum $F$.
The maximum concurrent flow problem motivated the study of multi-commodity flow \emph{without} capacities minimizing a mixed $\ell_{q,p}$-norm. 
For $q\to1,p\to\infty$ this is equivalent to maximum concurrent flow. 
Chen and Ye \cite{ChenY23} show that this problem is solvable to high-accuracy in almost-linear time for $1\le p\le 2 \le q$ with $p=\tilde{O}(1)$, $1/(q-1) = O(1)$ by reducing it to single-commodity flow.
Like our work, \cite{ChenY23} also provides new directions for efficient high-accuracy multi-commodity flow algorithms, circumventing lower bound arguments \cite{Itai78,DingKZ22}.

\paragraph{Linear programs and generalizations}
Previous efficient multi-commodity flow algorithms all rely on linear programming techniques such as interior point methods which are then combined with data structures to reduce the time per iteration. 
This technique of combining interior point methods with data structures is commonly used in various LP solvers \cite{Karmarkar84,Vaidya87,Vaidya89,LeeS15,CohenLS19,LeeSZ19,Brand20,JiangSWZ21,BrandLSS20,BrandLL+21} and also finds application in generalization such as semi-definite programs \cite{JiangKLPS20,HuangJST21,JiangNW22} and general convex optimization via cutting planes \cite{LeeSW15,JiangLSW20}.
An especially powerful data structure framework here is the inverse maintenance \cite{Sankowski05,LeeS15,CohenLS19,LeeSZ19,Brand20,BrandNS19,Brand21,JiangSWZ21}.

Converting high accuracy solutions for multi-commodity flow to exact solutions can be done with the same techniques as converting high accuracy solution of general LPs to exact solutions, since multi-commodity flow is just a special case.
\cite{DadushNV20} presents a scheme that takes $O(m)$ approximate solutions and converts them to an exact solution.
Another technique is to run the algorithm for small enough $\epsilon>0$ such that one can round to the nearest corner of the polytope representing the feasible solution space, see e.g.~\cite{Renegar88,LeeS14,KapoorV96} for a discussion on this.

\paragraph{Dynamic Expander Decomposition}
The improvements of our algorithm use data structures based on dynamic expander decompositions \cite{SaranurakW19,HuaKGW22}.
This is a fundamental tool in the area of dynamic graph algorithms, previously used to maintain properties of dynamic graphs such as connectivity \cite{WulffNilsen17,NanongkaiS17,NanongkaiSW17,ChuzhoyGLNPS20,JinS21,GoranciRST21}, distances \cite{Chuzhoy21,ChuzhoyS21,ChuzhoyK19,BernsteinGS21,BernsteinGS20}, approximate flows \cite{ChuzhoyGLNPS20,ChuzhoyK19,BernsteinGS21,GoranciRST21}, or sparsifiers \cite{BernsteinBGNSS022}.

Recently, dynamic expander decomposition has been used to develop data structure that accelerate single-commodity flow algorithms (i.e.~max flow, min-cost flow but also bipartite matching, transshipment etc.) \cite{BrandLL+21,GaoLP21,BrandGJLLPS22,ChenKL+22}. This line of work recently culminated in an almost-linear time algorithm for single-commodity flow \cite{ChenKL+22}.
Another important tool for single-commodity flow algorithms \cite{DaitchS08,Madry13,Madry16,CohenMSV17,AxiotisMV21,LeeS14,LiuS20,KathuriaLS20,BrandLL+21} are Laplacian system solvers which run in nearly-linear time and have an extensive history of research \cite{SpielmanT04,CohenKKPPRS18,KyngS16,KyngLPSS16,LeePS15,PengS14,CohenKMPPRX14,KelnerOSZ13,KoutisMP11,KoutisMP14}.

\subsection{Organization}
We start by giving a technical overview in \Cref{sec:overview}. There we present the main technical ideas, sketch the proof for our algorithm, and all tools and required data structures.
In \Cref{sec:IPM}, we analyze how the robust central path framework of \cite{CohenLS19,Brand20,JiangSWZ21} for general LPs behaves when solving k-commodity flows. In particular, we show that the size of required linear systems reduces and prove some additional guarantees that are required by our data structures.
These data structures are analyzed in \Cref{sec:vector}.
Finally, we combine everything to obtain our main results \Cref{thm:intro:mincost,thm:intro:throughput} in \Cref{sec:algorithm}.

\subsection{Preliminaries}

We use $\tilde{O}$ to hide $\polylog(m,k)$ factors and $\hO$ to hide sub-polynomial $m^{o(1)}$ factors.
We write $[n]$ for the interval $\{1,...,n\}$. For a set $I\subset[n]$ and $v\in\R^n$ we write $v_I$ for the sub-vector with entries whose index is in $I$.

Given two vectors $u,v\in\R^m$, all operations are defined elementwise, e.g.~$uv$ is an elementwise product so $(uv)_i = u_i\cdot v_i$ for all $i\in[m]$. Likewise $(u/v)_i = u_i/v_i$ for all $i\in[m]$. We also extend all scalar operations to be elementwise on vectors so, for example, $(\sqrt{v})_i = \sqrt{v_i}$ for all $i\in[m]$.
For an $\alpha\in\R$ we write $v+\alpha$ for adding $\alpha$ to each entry of $v$, so $(v+\alpha)_i = v_i+\alpha$ for all $i\in[m]$.

For vectors $x,s,d,g\in\R^m$ we write $\mX,\mS,\mD,\mG$ for the $m\times m$ diagonal matrices with the respective vectors on the diagonal, e.g., $\mX_{i,i} = x_i$ for all $i\in[m]$.

Given $\epsilon >0$, $\alpha,\beta\in\R$, we write $\alpha\approx_\epsilon\beta$ when $\exp(-\epsilon)\alpha\le\beta\le\exp(\epsilon)\alpha$. Note that for small $\epsilon$ we have that $\exp(\pm\epsilon)$ is roughly $(1\pm\epsilon)$ so the $\approx_\epsilon$ notation can be considered to reflect $(1\pm\epsilon)$ approximations.
The definition via the exponential function allows for the following transitive property: if $\alpha\approx_\epsilon\beta$ and $\beta\approx_\delta\gamma$ then $\alpha\approx_{\epsilon+\delta}\gamma$.
We also extend the notation to vectors, so $u\approx_\epsilon v$ means $u_i\approx_\epsilon v_i$ for all $i$.

\section{Technical Overview}
\label{sec:overview}
Our algorithm for solving multi-commodity flow comes from improving linear program solvers by developing data structures. One of these data structure problem is the ``heavy hitter problem''. Here one must preprocess a matrix $\mM$ and then support queries where for any given vector $v$, one must detect all large entries of the product $\mM v$. Solving the heavy hitter problem on sparse matrices is an open problem.
We make the first contribution for very structured sparse matrices as given in \Cref{def:intro:mcincidence} which we refer to as multi-commodity incidence matrices.
Since our data structure relies on the specific structure of the matrix, we must first prove that we indeed have a structure as in \Cref{def:intro:mcincidence}.
The first two subsection of this overview will recap how linear programs are solved via the robust central path method which will also prove the specific structure of the heavy hitter task.

In the first subsection \ref{sec:overview:2LP}, we define the LP structure of $k$-commodity flow.
We then give an overview for how to use the robust central path method to solve linear programs efficiently in subsection \ref{sec:overview:algebra}. There we also sketch why the $\tilde{O}((km)^\omega)$ complexity\footnote{For simplicity we hide $\log(U/\epsilon)$ factors in this overview.} from \cite{CohenLS19,LeeSZ19,Brand20,JiangSWZ21} can be reduced to $\tilde{O}(k^{2.5}(\sqrt{m}n^{\omega-1/2}+m^{3/2}))$ for $k$-commodity flow.
Finally, subsection \ref{sec:overview:singlecommodity} describes how to improve $\tilde{O}(k^{2.5}(\sqrt{m}n^{\omega-1/2}+m^{3/2}))$ to $\tilde{O}(k^{2.5}\sqrt{m}n^{\omega-1/2})$ time, by developing an efficient heavy hitter data structure on multi-commodity incidence matrices.
This last step uses techniques from single commodity flows as we show that heavy hitters on multi-commodity incidence matrices can be reduced to heavy hitters on classical incidence matrices, which in \cite{BrandLN+20} was solved using dynamic expander decomposition.

\subsection{$k$-Commodity LP}
\label{sec:overview:2LP}

Given a graph $G=(V,E)$, let $\mB\in\R^{E\times V}$ be the edge-vertex incidence matrix. 
Let $u\in\R^E_{>0}$ be edge capacities, $c_1,...,c_k\in\R^E$ be the costs and $d_1,...d_k \in \R^V$ be the demand vectors for the $k$ commodities. Then, we can write the min-cost variant of $k$-commodity flow as the LP
\begin{align*}
\min_{x_1,...,x_k} \sum_{i=1}^k c_i^\top x_i &\text{ subject to} \\
\mB^\top x_i = d_i &\text{ for all } 1\le i\le k,\\
\sum_{i=1}^k x_i \le u,&~~~x_1,...,x_k\ge0%
\end{align*}
where $x_1,...,x_k \in \R^E$ are the flows for the $k$ commodities.
This can be written as an LP in standard form by introducing a slack variable $x_{k+1}\ge0$ with $\sum_{i=1}^{k+1}x_i=u$ and writing $x=(x_1,...,x_{k+1})\in\R^{(k+1)E}$. The primal LP $\cP$ (and its dual $\cD$) are given by
\begin{align*}
\begin{array}{r}
    (\cP)~~\min_x c^\top x  \\
    \cB^\top x = d, \\
    x\ge0
\end{array}
&~~~~
\begin{array}{r}
    (\cD)~~\max_y d^\top y  \\
    \cB y +s = c, \\
    s\ge0
\end{array}\notag\\
\text{ where }
\cB := \begin{bmatrix}
    \mB &  & 0 & \mI \\
     & \ddots &  & \vdots \\
    0 &   & \mB & \mI \\
    0 & ... & 0 & \mI
\end{bmatrix}&
\text{ and }
\begin{array}{cc}
    d:=(d_1,...,d_k, u)\in\R^{kV+E}\\
    c:=(c_1,...,c_k,0)\in\R^{(k+1)E}
\end{array} %
\end{align*}
Using \cite{CohenLS19,Brand20,JiangSWZ21} to solve this LP would take $\tilde{O}((km)^\omega)$ time since $\cB$ is of size $(k+1)m \times (kn+m)$.
Kapoor and Vaidya \cite{KapoorV96} solve this LP in $\tilde{O}(k^{2.5}\sqrt{m}n^2)$ time
while Lee and Sidford solve it in $\tilde{O}((kn+m)^{2.5})$ time \cite{LeeS14,LeeS15}.
These are the fastest algorithms for solving $k$-commodity flow to high accuracy and none of them use the graph structure of $\mB$.
In particular, if we were to replace $\mB$ by any other $m\times n$ matrix $\mA$ with $\nnz(\mA)$ non-zero entries, \cite{CohenLS19},\cite{KapoorV96},\cite{LeeS14,LeeS15} would still run in $\tilde{O}((km)^\omega)$, $\tilde{O}(\sqrt{km}(k^2n^2+k\nnz(\mA)))$, or $\tilde{O}(\sqrt{kn+m}(k\nnz(\mA)+(kn+m)^2)$ time.

Let us call this type of LP a ``\emph{$k$-commodity LP}'' since it is an LP with $k$ commodities $x_1,...,x_k$ that do not necessarily represent a flow as $\mA$ does not have to be an incidence matrix.
\begin{definition}\label{def:overview:commodityLP}
Given $m\times n$ matrix $\mA$, $k$ demand and cost vectors $d_1,...,d_k\in\R^{n}$, $c_1,...,c_n\in\R^{m}$
we call the following primal and dual linear programs a $k$-commodity LP. 
\begin{align*}
\begin{array}{r}
    (\cP)~~\min_x c^\top x  \\
    \cA^\top x = d, \\
    x\ge0
\end{array}
~
\begin{array}{r}
    (\cD)~~\max_y d^\top y  \\
    \cA y +s = c, \\
    s\ge0
\end{array} %
\text{ where }
\cA := \begin{bmatrix}
    \mA &  & 0 & \mI \\
     & \ddots &  & \vdots \\
    0 &   & \mA & \mI \\
    0 & ... & 0 & \mI
\end{bmatrix}&
\text{ and }
\begin{array}{cc}
    d:=(d_1,...,d_k, u)\in\R^{kn+m}\\
    c:=(c_1,...,c_k,0)\in\R^{(k+1)m}
\end{array}
\end{align*}
\end{definition}

\subsection{Robust IPM and Algebraic Techniques}
\label{sec:overview:algebra}
In this subsection we show how techniques from \cite{KapoorV96} can be combined with the recent robust interior point framework and data structures from \cite{CohenLS19,Brand20,JiangSWZ21} to obtain a faster algorithm for $k$-commodity LPs.
This leads to an algorithm with complexity $\tilde{O}(k^{2.5}m^{1/2}(n^{\omega-1/2}+\nnz(\mA)))$ for sparse $\mA$ with only $O(1)$ non-zeros per row.
This first speed up relies only on algebraic techniques and works for any sparse $k$-commodity LP. 
In the later \Cref{sec:overview:singlecommodity}, we show how to accelerate this algorithm further via graph based single-commodity flow techniques, when $\mA$ is an incidence matrix. 
This then leads to our fast $k$-commodity flow algorithms \Cref{thm:intro:throughput,thm:intro:mincost}.

\paragraph{Robust Central Path}
Let us start with a quick recap of the robust central path method which has led to many improvements in LP solvers \cite{LeeSZ19,BrandLSS20,JiangSWZ21,Brand20} and single-commodity flow algorithms \cite{BrandLN+20,BrandLL+21,GaoLP21,BrandGJLLPS22}.
The robust central path method is a variant of the classical central path method for solving LPs. 
Here, one has two iterates $x$ (a feasible but not optimal solution to the primal LP $\cP$ in \Cref{def:overview:commodityLP}) and $s$ (the slack of the dual LP $\cD$). 
There are $\tilde{O}(\sqrt{km})$ iterations and in each iteration, both $x$ and $s$ are moved a bit towards the optimal solution.
In the robust central path method, the movement of $x$ and $s$ are given by $x\leftarrow x+\delta_x$, $s\leftarrow s+\delta_s$ where
\begin{align}
\delta_s =&~ \cA(\cA^\top \omX \omS^{-1} \cA)^{-1} \cA^\top \omS^{-1} \og \label{eq:overview:deltas}\\
\delta_x =&~ \omS^{-1} \og - \omX\omS^{-1}\delta_s. \notag
\end{align}
where $\omX,\omS$ are diagonal matrices with $\omX_{i,i}\approx x_i, \omS_{i,i}\approx s_i$ for all $i$, and $\og$ is some vector specified by the robust central path method. Here $\omX,\omS,\og$ may change from one iteration to the next.

\paragraph{Dimension Reduction}
Note that $(\cA^\top \omX \omS^{-1} \cA)$ in \eqref{eq:overview:deltas} is an $(m+kn)\times(m+kn)$ matrix (by definition of $\cA$ in \Cref{def:overview:commodityLP}), so one would expect at least $\poly(m)$ time per iteration to multiply a vector with it.
However, using the $k$-commodity structure of $\cA$, one can reduce the dimension of the linear system onto a smaller inverse of size $kn\times kn$ by taking the Schur-complement. This was first observed in \cite{KapoorV96}. 
We briefly sketch how taking the Schur-complement reduces the dimension. We have
\begin{align}
    \cA^\top \omX \omS^{-1} \cA
    = 
    \left[\begin{array}{ccc|c}
    \mA^\top \mD_1 \mA &  &   0   &   \mA^\top \mD_1\\
    & \ddots & &\vdots \\
    0   &  & \mA^\top\mD_k\mA  &   \mA^\top\mD_k\\
    \hline 
    \mD_1\mA  & \hdots&  \mD_k\mA   &    \mD_\Sigma
    \end{array}\right] \label{eq:overview:hessian}
\end{align}
where for ease of notation we defined $\mD_i = \omX_i\omS_i^{-1}$ where $\omX_i,\omS_i$ are the $m\times m$ submatrices with rows in $[(i-1)m+1, im]$ for $i=1,...,k+1$ , and $\mD_{\Sigma}=\sum_{i=1}^{k+1} \mD_i$.
By taking the Schur-complement (see \Cref{lem:structure:inverse}) the inverse of \eqref{eq:overview:hessian} is given by
\begin{align}
&~
\left[\begin{array}{ccc|c}
\mI&&&0\\
&\ddots&&\vdots\\
&&\mI&0\\
\hline
-\mD_\Sigma^{-1}\mD_1\mA&\cdots&-\mD_\Sigma^{-1}\mD_k\mA&\mI
\end{array}\right]
\left[\begin{array}{ccc|c}
    &&&0\\
    &\mE^{-1}& &\vdots\\
    &&&0\\
    \hline
    0&\cdots&0&\mI
\end{array}\right]
\left[\begin{array}{ccc|c}
    \mI&&&-\mA^\top\mD_1\mD_\Sigma^{-1}\\
    &\ddots&&\vdots\\
    &&\mI&-\mA^\top\mD_k\mD_\Sigma^{-1}\\
    \hline
    0&\cdots&0&\mD_\Sigma^{-1}
\end{array}\right] \label{eq:overview:inverse} \\
&~\text{where }
\mE :=
    \left[\begin{array}{ccc}
    \mA^\top \mD_1 \mA &  &   0  \\
     & \ddots & \\
    0   &  & \mA^\top\mD_k\mA 
    \end{array}\right]-
    \left[\begin{array}{c}
    \mA^\top \mD_1\\
    \vdots \\
    \mA^\top\mD_k\\
    \end{array}\right]
    \mD_\Sigma^{-1}
    \left[\begin{array}{ccc}
    \mD_1\mA  &  \cdots & \mD_k\mA
    \end{array}\right] \notag
\end{align}
Note that here only the smaller $kn\times kn$ matrix $\mE$ must be inverted\footnote{%
If $\mA$ is an incidence matrix, then this matrix $\mE$ was referred to as a ``multi-commodity Laplacian'' in \cite{KyngZ20}, and despite graph structure, solving a linear system in $\mE$ is as hard as solving a general linear system. 
Further, after reordering rows and columns, $\mE$ is exactly $\mM^\top\mM$ for $\mM$ as in \Cref{def:intro:mcincidence} where $h^{(u,v)}_i$ from \Cref{def:intro:mcincidence} is the diagonal entry of $\mD_i \in \R^{m\times m}$ corresponding to edge $(u,v)$.
} 
(since $\mD_\Sigma$ is a diagonal matrix and thus trivial to invert).
Further, multiplying the vector $\cA^\top\omS^{-1}\og$ with $(\cA^\top\omX\omS^{-1}\cA)^{-1}$ as in \eqref{eq:overview:deltas} reduces to computing the following expression (proven in \Cref{lem:structure:delta}):
\begin{align*}
    w =&~ \mD_\Sigma^{-1} \sum_{i=1}^{k+1} \omS_i^{-1} g_i &%
    \begin{pmatrix}v_1\\\vdots\\v_k\end{pmatrix} =&~
    \mE^{-1}
    \underbrace{\begin{pmatrix}
    \mA^\top (\omS_1^{-1} g_1 -\mD_1 w)\\
    \vdots\\
    \mA^\top (\omS_k^{-1} g_k -\mD_k w)
    \end{pmatrix}}_{=:u}.
\end{align*}
Assume for now that $\mA$ has sparse rows with $O(1)$ non-zero entries (as would be the case for $k$-commodity flow where $\mA$ is an incidence matrix).
If an entry of $\omX,\omS$, or $\og$ changes, then only $O(k^2)$ entries in $\mE$ and $u$ change,
because of sparsity of rows of $\mA$.
Using a dynamic linear system data structure (\Cref{lem:inverse_maintenance}), one can maintain the solution of this linear system $\mE^{-1}u$ efficiently under entry updates to $\mE$ and $u$.
In particular, an amortized complexity of $\tilde{O}(k^2(n^{1+\mu} + n^{\omega(1,1,\mu) - \mu/2} + n^{\omega-1/2}))$ per iteration of the central path method is possible for any trade-off parameter $0\le\mu\le1$.
For current bounds on $\omega$, this is just $\tilde{O}(k^2 n^{\omega-1/2})$ amortized time per iteration.
The only remaining problem is how to compute the product with the matrix on the left of $\mE^{-1}$ in \eqref{eq:overview:inverse} and the product with $\cA$ in definition of $\delta_s$ in \eqref{eq:overview:deltas}.

For this, we prove in \Cref{sec:IPM} that the robust central path method (i.e.~\eqref{eq:overview:deltas}) can be rewritten as follows.
Compute $\delta_s=(\delta_s^1,...,\delta_s^{k+1})$ where for $v_{k+1}:=0$
\begin{align}
\delta_s^{(i)} =&~ w + \mA v_i - \sum_{j=1}^k \mD_j\mD_\Sigma^{-1}\mA v_j  ~~~\text{ for } i =1,...,k+1\label{eq:overview:modifieddeltas}
\end{align}
So after computing the vectors $(v_1,...,v_k)=\mE^{-1}u$ in $\tilde{O}(k^2n^{\omega-1/2})$ amortized time, we are only left with multiplying the resulting vectors with computing \eqref{eq:overview:modifieddeltas}.
In general, this would take $\nnz(k^2\mA)$ time, but if $\mA$ is an incidence matrix, we can use data structures from single-commodity flow algorithms \cite{BrandLN+20,BrandLL+21} based on the expander decomposition framework.
This is outlined in the next subsection.

\paragraph{Heavy Hitters} 
A common technique when solving linear programs via the robust central path method, is to use ``heavy hitter'' data structures. 
(See e.g.~\cite{BrandLSS20,BrandLL+21,BrandLN+20,GaoLP21,BrandGJLLPS22}.)
Note that we only need $\omX,\omS$ in \eqref{eq:overview:deltas} with $\omX_{i,i} \approx x_i$ and $\omS_{i,i} \approx s_i$ for all $1\le i\le (k+1)m$.
Thus we do not actually need to compute vectors $x,s$ explicitly.
Further, if an entry $s_i$ did not change much from one iteration to the next, we can reuse the old value of $\omS_{i,i}$ in the next iteration. Only when $s_i$ changes sufficiently, do we need to update $\omS_{i,i}$.
This leads to the following data structure problem:
in each iteration find for which $1\le i\le(k+1)m$ the entries $|(\delta_s)_i| > \epsilon \omS_{i,i}$
for some parameter $\epsilon \in (0,1]$.
That is, find the ``heavy hitters'' of the vector $\omS^{-1}\delta_s$. (And likewise $\omX^{-1}\delta_x$ but by \eqref{eq:overview:deltas} this is almost the same problem as $\omS^{-1}\delta_s$.)
This allows for a speed-up of the algorithm, because not the entire vector $\delta_s$ must be computed.
We explain in the next \Cref{sec:overview:singlecommodity} how this data structure task can be solved efficiently using graph techniques when matrix $\mA$ is an incidence matrix.

\subsection{Single-Commodity Techniques}
\label{sec:overview:singlecommodity}

As outlined in the previous subsection \ref{sec:overview:algebra}, the remaining task for obtaining an efficient multi-commodity flow algorithm is to solve a ``heavy hitter'' data structure problem.
The task is to find the large entries of the vector $\omS^{-1}\delta_s$ (as defined in \eqref{eq:overview:modifieddeltas}).
The central path method guarantees that $\omX\omS \approx t\mI$ for some $t\in\R{>0}$, so finding entries larger than some $\epsilon$ in $\omS^{-1}\delta_x$ can be done by finding entries larger than
$\epsilon \sqrt{t}$ in $\mD^{1/2}\delta_s$ (reminder: $\mD=\omX~\omS^{-1}$) by
$$
\mD^{1/2}\delta_s = \omX^{1/2}\omS^{-1/2} \delta_s = (\omX\omS)^{1/2}~\omS^{-1}\delta_s \approx \sqrt{t}\cdot \omS^{-1}\delta_s.
$$
By definition of $\delta_s$ in \eqref{eq:overview:modifieddeltas}\footnote{For simplicity, we will from now on ignore the vector $w$ in \eqref{eq:overview:modifieddeltas}. Since the vector $w$ is explicitly given, it's easy to check if $\mD^{1/2}w$ has large entries. It's the sum of products with $\mA$ for which finding the large entries is the bottleneck.}, that means we try to find large entries of
\begin{align}
\mD^{1/2}_i\left( \mA v_i - \sum_{j=1}^k \mD_j\mD_\Sigma^{-1}\mA v_j\right)\text{ for each }i=1,...,k+1
\label{eq:overview:explicitsum}
\end{align}
(where $\mD_i$ is the $i\th$ $m\times m$ diagonal subblock of $\mD$.)

We remark that finding large entries of \eqref{eq:overview:explicitsum} is precisely the task we previously described in \Cref{sec:intro:techniques} (for a matrix as in \Cref{def:intro:mcincidence}).
This is because \eqref{eq:overview:explicitsum} can be phrased as finding the large entries of the following matrix vector product where the matrix is a multi-commodity incidence matrix:
\begin{align*}
    \mD^{1/2}\delta_s
    =
    \begin{bmatrix}
        \mD_1^{1/2} \delta_s^{(1)} \\
        \vdots \\
        \mD_{k+1}^{1/2} \delta_s^{(k+1)}
    \end{bmatrix}
    =
    \left(\begin{bmatrix}
        \mD_1^{1/2}\mA & & 0\\
        & \ddots & \\
        0 & & \mD_{k+1}^{1/2}\mA \\
    \end{bmatrix}
    -
    \begin{bmatrix}
        \mD_1^{1/2}\mD_1 \mD_\Sigma^{-1} \mA & \hdots & \mD_1^{1/2}\mD_{k+1}\mA \\
        \vdots & \ddots & \vdots\\
        \mD_{k+1}^{1/2}\mD_1 \mD_\Sigma^{-1} \mA & \hdots & \mD_{k+1}^{1/2}\mD_{k+1}\mA \\
    \end{bmatrix}
    \right)
    \begin{bmatrix}
        v_1 \\
        \vdots \\
        v_{k+1} \\
    \end{bmatrix}
\end{align*}
After reordering rows and columns, the matrix in above equation is of form \Cref{def:intro:mcincidence}\footnote{
With $h^{(u,v)}_i$ in \Cref{def:intro:mcincidence} being the diagonal entry of $\mD_i\in\R^{m\times m}$ that corresponds to edge $(u,v)\in E$.}, i.e.~it is a multi-commodity incidence matrix.

We must now solve the heavy hitter problem on this sparse matrix with $O(k)$ non-zero entries per row. %
An intuitive idea would be to %
split the heavy hitter problem on the multi-commodity incidence matrix into smaller tasks on classical incidence matrices,
e.g.~instead of finding $1\le i \le k+1$ and $1\le \ell \le m$ where
\begin{align}
|(\mD_i^{1/2}\delta_s^{(i)})_\ell|=
\left|\left(\mD_i^{1/2}\left(\mA v_i
    - \sum_{j=1}^{k+1}\mD_j\mD_\Sigma^{-1} \mA v_j \right)
    \right)_\ell\right| > \epsilon \sqrt{t}
    \label{eq:overview:deltas_condition}
\end{align}
we search for indices $1\le i \le k+1$ and $1\le \ell\le m$ where
\begin{align}
    \left|\left(\mD_i^{1/2} \mA v_i\right)_\ell\right| >&~ \frac{\epsilon}{k+1}\sqrt{t}
    ~\text{ or }\notag\\
    \left|\left(\mD_i^{1/2}\mD_j\mD_\Sigma^{-1} \mA v_j\right)_\ell\right| >&~ \frac{\epsilon}{k+1}\sqrt{t}
    ~\text{ for any }1\le j \le k+1\label{eq:overview:modified_deltas_condition}
\end{align}
The problem with this approach is that it is not clear how many indices will be returned this way.
Usually, the number of indices that satisfy \eqref{eq:overview:deltas_condition} is bounded by $O(\|\mD^{1/2}_j\delta_s^j\|_2/(\epsilon^2 t))$, and this norm is bounded by properties of the central path method.
However, no such bound is given on the norms of the vectors in \eqref{eq:overview:modified_deltas_condition},
so we cannot bound the number of returned indices.
In particular, it could be that for some $i,\ell$ there are two different $j$ where \eqref{eq:overview:modified_deltas_condition} is satisfied, but one entry is a large positive value and the other is a large negative value. So in \eqref{eq:overview:deltas_condition} the respective entry might still be small because of cancellation.
As returning an index takes at least $O(1)$ time,
we cannot bound the complexity for finding large entries in \eqref{eq:overview:modified_deltas_condition}.

To solve heavy hitter task on \eqref{eq:overview:deltas_condition}, we find that there is actually a decomposition of the multi-commodity incidence matrix into smaller classical incidence matrices, where we can guarantee that such cancellations happen rarely.

Let us write $\delta_s^{(i)}$ in a slightly different form:
\begin{align*}
\mD_i^{1/2}\delta_s^{(i)} 
= 
\mD_i^{1/2}\left(\mA v_i - \sum_{j=1}^{k+1} \mD_j \mD_\Sigma^{-1} \mA v_j\right)
=
\sum_{j=1}^{k+1} \mD_i^{1/2}\mD_j\mD_\Sigma^{-1}\mA (v_i-v_j)
\end{align*}
(Here we used that $\mD_\Sigma=\sum_{i=1}^{k+1}\mD_i$.)
Thus, we can replace the conditions in \eqref{eq:overview:modified_deltas_condition} by the following conditions
\begin{align}
    \left|\left(\mD_i^{1/2}\mD_j\mD_\Sigma^{-1}\mA (v_i- v_j) \right)_\ell\right| > \frac{\epsilon}{k+1} \sqrt{t} \text{ for any } 1\le j \le k+1
    \label{eq:overview:modified_deltas_condition_2}
\end{align}
Unlike the terms in \eqref{eq:overview:modified_deltas_condition}, we can bound the norms of above terms. 
\begin{lemma}[{Shortened version of \Cref{lem:vector:simplifiednorm}}]\label{lem:overview:normbounds}
    For any $v_1, ...,v_{k+1} \in \R^n$, 
    $d_1,...,d_{k+1}\in\R^{m}_{>0}$ 
    let $d_\Sigma=\sum_{i=1}^{k+1}d_i$. Write $\mD_i$ and $\mD_\Sigma$ for the diagonal matrices with $d_i$ and $d_\Sigma$ on the diagonal.
    Then
    \begin{align*}
\sum_{i=1}^{k+1}\sum_{j=1}^{k+1}
\left\lVert\mD_i^{1/2}\mD_j\mD_\Sigma^{-1}\mA(v_i-v_j)\right\rVert_2^2
\le 4\cdot
\sum_{i=1}^{k+1} \left\|\mD^{1/2}_i \sum_{j=1}^{k+1} \mD_j\mD_\Sigma\mA(v_i - v_j)\right\|_2^2
\end{align*}

\end{lemma}
This norm upper bound from \Cref{lem:overview:normbounds} is somewhat unexpected: note that tools such as triangle inequality would give an inequality in the opposite direction. A proof is given in \Cref{sec:vector} via \Cref{lem:vector:simplifiednorm}.

Since the number of large entries in \eqref{eq:overview:modified_deltas_condition_2} can be bounded w.r.t~the norms in \Cref{lem:overview:normbounds}, we can bound how many entries are large. At the same time, we have %
\begin{align*}
\sum_{i=1}^{k+1} \left\|\mD^{1/2}_i \sum_{j=1}^{k+1} \mD_j\mD_\Sigma\mA(v_i - v_j)\right\|_2^2 %
=&~
\sum_{i=1}^{k+1} \|\mD^{1/2}_i \delta_s^{(i)}\|_2^2
= 
\|\mD\delta_s\|_2^2
\end{align*}
which is bounded by guarantees of the central path method. So we never return too many large entries in \eqref{eq:overview:modified_deltas_condition_2}.
At last, we use the heavy hitter data structure for classical incidence matrices from \cite{BrandLN+20}, originally developed for the purpose of solving single-commodity flows.

\paragraph{Summary}
In summary, we can solve the heavy hitter problem on multi-commodity incidence matrices $\mM$ by splitting it into $O(k^2)$ instances on classical incidence matrices.
This is quite different form the inverse subroutine: solving multi-commodity Laplacian $\mM^\top\mM$ are unlikely be efficiently reducible to regular Laplacian systems \cite{KyngZ20} unless there is some breakthrough in sparse linear system solving. 
This demonstrates that while any LP can be reduced to multi-commodity flow \cite{Itai78,DingKZ22}, 
some aspects of solving multi-commodity flow on dense graphs can actually be reduced to the single commodity case and accelerated using graph techniques.

\section{Interior Point Method}
\label{sec:IPM}
The robust central path framework is a common tool in the development of efficient linear program solvers and single commodity flow algorithms \cite{CohenLS19,LeeSZ19,Brand20,JiangSWZ21,BrandLL+21,BrandLN+20,BrandLSS20}.
In this section we analyze the robust central path framework when applied to a $k$-commodity LP.
When used as blackbox, the robust central path framework would have to solve an $O(km)\times O(km)$ dimensional linear system in each iteration. 
In \Cref{sec:IPM:structure} we show that this problem can be reduced to an $O(kn)\times O(kn)$ dimensional linear system. That is, we prove the claims from \Cref{sec:overview:algebra} that the inverse is of shape as described in \eqref{eq:overview:inverse}.
We also prove that our steps $\delta_s$ to the slack of the dual are of form \eqref{eq:overview:modifieddeltas}. Especially the form of the steps is important because only for very structured sparse matrices can we maintain the primal $x$ and slack $s$ of the dual with small amortized time per iteration.
A data structure for maintaining $x$, $s$ is given in \Cref{sec:vector}.

\subsection{Structure of IPM for $k$-Commodity LP}
\label{sec:IPM:structure}
As a reminder, we defined a $k$-commodity LP (\Cref{def:overview:commodityLP}) as an LP of the form
\begin{align*}
\min_{x_1,...,x_{k+1}\ge 0} \sum_{i=1}^k c_i^\top x_i& \text{ subject to} \\
\mA^\top x_i = d_i& \text{ for } i=1,...,k\\
\sum_{i=1}^{k+1} x_i = u&
\end{align*}
where $\mA \in \R^{m\times n}$ for $m\ge n$ is full-rank (i.e.~rank $n$).\footnote{If we let $\mA$ be an edge vertex incidence matrix to model a $k$-commodity flow instance, then this full-rank assumption is not satisfied. 
We describe in \Cref{sec:algorithm:initial} how to fix this issue.}
\Cref{thm:commodity_ipm} shows that the central path method given by \Cref{alg:commodity_ipm} can be used to solve this linear program.

\begin{theorem}\label{thm:commodity_ipm}
    Consider a $k$-commodity LP for which we call $\textsc{Solve}(\mA \in \R^{m\times n}, x^\init,s^\init \in \R^{(k+1)m}_{>0},t^\init,t^\target\in\R_{>0})$ (\Cref{alg:commodity_ipm})
    and let $x^\final$, $s^\final$ be the output. %
    If input $x^\init,s^\init$ is feasible with $\Phi(\frac{x^\init s^\init}{t^\init})\le 16 \cdot (k+1)m$, then
    \begin{itemize}[nosep]
    \item the output is feasible with $|\frac{x^\final s^\final}{t^\target}-1| \le 1/16$,
    \item the number of iterations is $O(\sqrt{km} \lambda |\log(t^\init/t^\target)|)$,
    \item and in each step we have
    $\|\mS^{-1} \delta_s\|_2,\|\mX^{-1}\delta_x\|_2 \le \lambda/16$, and $\|xs/t - 1\|_\infty \le 1/16$.
    \end{itemize}
\end{theorem}

\begin{algorithm2e}[t!]
\caption{Robust central path algorithm, applied to a $k$-commodity LP. \label{alg:commodity_ipm}}
\SetKwProg{Proc}{procedure}{}{}
\SetKwProg{Params}{parameters}{}{}
\Params{}{
	$\lambda = 16 \log 40 \sqrt{(k+1)m}$,
	$h = 1/(128\lambda\sqrt{(k+1)m})$,
	$\Phi(v) := \sum_{i\in[(k+1)m]} (\exp(\lambda v_i) + \exp(-\lambda v_i))/2$.
}
\Proc{$\textsc{Solve}(\mA\in\R^{m\times n}, x,s \in \R^{(k+1)m}_{>0}, t^\init > 0, t^\target > 0)$}{
\While{$t \neq t^\target$}{
Pick $\ox\approx_{1/48} x$, $\os \approx_{1/48} s$, $\|\ov-xs/t\|_\infty \le \frac{1}{48\lambda}$\\
Let $t' = \max(t/(1+h), t^\target)$\\
Let $g = -\nabla \Phi(\ov) = (g_1, g_2, \ldots, g_{k+1})$ where $g_i \in \R^m$ for $i=1,2,\ldots,k+1$. \\
For sub-vectors $(\ox_1, \ox_2, ..., \ox_{k+1})=\ox,(\os_1, \os_2, ..., \os_{k+1})=\os\in\R^{(k+1)m}$, let $\mD_i = \omX_i\omS^{-1}_i$ for $i=1,...,k+1$ and let $\mD_\Sigma = \sum_{i=1}^{k+1} \mD_i$. Let
\begin{align}
    \mE := 
    \left[\begin{array}{ccc}
    \mA^\top \mD_1 \mA &  &   0  \\
     & \ddots & \\
    0   &  & \mA^\top\mD_k\mA 
    \end{array}\right]-
    \left[\begin{array}{c}
    \mA^\top \mD_1\\
    \vdots \\
    \mA^\top\mD_k\\
    \end{array}\right]
    \mD_\Sigma^{-1}
    \left[\begin{array}{ccc}
    \mD_1\mA  &  \cdots & \mD_k\mA
    \end{array}\right]
    \label{eq:ipm:defineE}
\end{align}\\
Compute $\delta_x = (\delta_x^1,...,\delta_x^{k+1}), \delta_s=(\delta_s^1,...,\delta_s^{k+1})$ where
\begin{align*}
    w =&~ \mD_\Sigma^{-1} \sum_{i=1}^{k+1} \omS_i^{-1} g_i &%
    \begin{pmatrix}
    v_1\\ 
    \vdots\\
    v_k
    \end{pmatrix} 
    =&~
    \mE^{-1}
    \begin{pmatrix}
    \mA^\top (\omS_1^{-1} g_1 -\mD_1w)\\
    \vdots \\
    \mA^\top (\omS_k^{-1} g_k -\mD_kw)
    \end{pmatrix}
\end{align*}
\begin{align*}
\delta_s^i =&~ \left(w+\mA v_i-\sum_{j=1}^k \dfrac{d_j}{d_\Sigma}\mA v_j\right) \frac{t'}{32\lambda\|g\|_2} \text{ for } i=1,...,k\\
\delta_s^{k+1} =&~ \left(w-\sum_{j=1}^k \dfrac{d_j}{d_\Sigma}\mA v_j\right) \frac{t'}{32\lambda\|g\|_2} \\
\delta_x^i =&~ \omS^{-1}_i \frac{t'}{32\lambda} \frac{g_i}{\|g\|_2} - \omX_i\omS_i^{-1} \delta_s^i \text{ for } i = 1,..,k+1.
\end{align*} \label{line:commodityipm:deltaxdeltas} \\
$x \leftarrow x + \delta_x$,
$s \leftarrow s + \delta_s$,
$t \leftarrow t'$
}
\Return $x, s$
}
\end{algorithm2e}

\noindent
A $k$-commodity LP can also be written as follows
\begin{align}
\begin{array}{cc}
     \min_{x\ge 0} c^\top x \text{ subject to} \\
\cA^\top x = d
\end{array}
\text{ where }
\cA := \left[\begin{array}{ccc|c}
       \mA & & 0 & \mI \\
        & \ddots & & \vdots \\
        0 & & \mA & \mI \\
        \hline
        0 & \cdots & 0 & \mI
    \end{array}\right],
c := \begin{bmatrix}
    c_1 \\
    \vdots \\
    c_k \\
    0
\end{bmatrix},
d := \begin{bmatrix}
    d_1 \\
    \vdots \\
    d_k \\
    u
\end{bmatrix}
\label{eq:ipm:standardform}
\end{align}
We prove \Cref{thm:commodity_ipm} by showing that \Cref{alg:commodity_ipm} performs the same steps $x\leftarrow x+\delta_x, s\leftarrow s+\delta_s$ as the classic robust central path method for LPs in standard form (\Cref{alg:robust_ipm}), when applied to \eqref{eq:ipm:standardform}.
The correctness of the classic robust central path method (\Cref{alg:robust_ipm}) is given by \Cref{lem:robust_ipm}.

\begin{lemma}[\cite{LeeV21}]\label{lem:robust_ipm}
    Consider a linear program $\min c^\top x$ subject to $\mA^\top x = b$, $x\ge0$ for $n\times d$ matrix $\mA$. 
    Assume we call $\textsc{Solve}(\mA, x^\init,s^\init,t^\init,t^\target)$ (\Cref{alg:robust_ipm})
    and let $x^\final$, $s^\final$ be the output. %
    If input $x^\init,s^\init$ is feasible with $\Phi(\frac{x^\init s^\init}{t^\init})\le 16 n$, then
    \begin{itemize}[nosep]
    \item the output is feasible with $|\frac{x^\final s^\final}{t^\target}-1| \le 1/16$,
    \item the number of iterations is $O(\sqrt{n} \lambda |\log(t^\init/t^\target)|)$,
    \item and in each step we have
    $\|\mS^{-1} \delta_s\|_2,\|\mX^{-1}\delta_x\|_2 \le \lambda/16$, and
    $\|xs/t - 1\|_\infty \le 1/16$
    \end{itemize}
\end{lemma}

\begin{algorithm2e}[t]
\caption{Robust central path algorithm. \label{alg:robust_ipm}}
\SetKwProg{Proc}{procedure}{}{}
\SetKwProg{Params}{parameters}{}{}
\Params{}{
	$\lambda = 16 \log 40 \sqrt{n}$,
	$h = 1/(128\lambda\sqrt{n})$,
	$\Phi(v) := \sum_{i\in[n]} (\exp(\lambda v_i) + \exp(-\lambda v_i))/2$.
}
\Proc{$\textsc{Solve}(\cA \in \R^{n\times d}, x \in \R^n_{>0}, s \in \R^n_{>0}, t^\init > 0, t^\target > 0)$}{
\While{$t \neq t^\target$}{
Pick $\ox\approx_{1/48} x$, $\os \approx_{1/48} s$, $\|\ov-xs/t\|_\infty \le \frac{1}{48\lambda}$\\
Let $t' = \max(t/(1+h), t^\target)$, $g = -\nabla \Phi(\ov)$ \\
Compute $\delta_x, \delta_s$ where
\begin{align}
\delta_s =&~ \cA(\cA^\top\omX\omS^{-1}\cA)^{-1}\cA^\top\omS^{-1}\frac{t'}{32\lambda} \frac{g}{\|g\|_2}  \label{eq:ipm}\\
\delta_x =&~ \omS^{-1}\frac{t'}{32\lambda} \frac{g}{\|g\|_2} - \omX\omS^{-1}\delta_s \notag
\end{align}
$x \leftarrow x + \delta_x$,
$s \leftarrow s + \delta_s$,
$t \leftarrow t'$
}
\Return $x, s$
}
\end{algorithm2e}

To show that \Cref{alg:commodity_ipm} performs a similar computation as \eqref{eq:ipm} in \Cref{alg:robust_ipm}, we start by observing the following \Cref{lem:structure:hessian}.

\begin{fact}\label{lem:structure:hessian}
For given $\ox=(\ox_1, ..., \ox_{k+1}),\os=(\os_1, ..., \os_{k+1})\in\R^{(k+1)m}$, let $\mD_i = \omX_i\omS^{-1}_i$ for $i=1,...,k+1$.
Further, let $\mD_\Sigma = \sum_{i=1}^{k+1} \mD_i$. Then
    \begin{align*}
    \cA^\top \omX \omS^{-1} \cA
    = 
    \left[\begin{array}{ccc|c}
    \mA^\top \mD_1 \mA   &      &  0  &   \mA^\top \mD_1\\
     & \ddots &  &  \vdots \\
    0  &   &   \mA^\top\mD_k\mA  &   \mA^\top\mD_k\\
    \hline 
    \mD_1\mA  &  \cdots & \mD_k\mA   &    \mD_\Sigma
    \end{array}\right]
    \end{align*}
\end{fact}

We can rewrite the inverse of this matrix using the following lemma.
This is just the Schur complement of a block matrix. For completeness, a proof is given in \Cref{sec:appendix:schur}.

\begin{restatable}{lemma}{schurcomplement}
\label{lem:structure:psd_inverse}
Suppose we are given a positive definite block matrix of form
$$\begin{pmatrix}A&B\\C&D\end{pmatrix}\in\mathbb{R}^{(m+n)\times(m+n)}.$$
Then, $D\in\mathbb{R}^{n\times n}$ and $E=A-BD^{-1}C$ are positive definite and hence invertible. 
Furthermore,
\begin{align}
\begin{pmatrix}A&B\\C&D\end{pmatrix}^{-1}
&=
\begin{pmatrix}I&0\\-D^{-1}C&I\end{pmatrix}
\begin{pmatrix}E^{-1}&0\\0&I\end{pmatrix}
\begin{pmatrix}I&-BD^{-1}\\0&D^{-1}\end{pmatrix}
\label{eq:structure:psd_inverse}
\end{align}
\end{restatable}

The inverse of $\cA^\top\omX\omS^{-1}\cA$ can now be written as follows, which requires only the inversion of an $O(n)\times O(n)$ matrix.

\begin{lemma}\label{lem:structure:inverse}
For given $\ox=(\ox_1, ..., \ox_{k+1}),\os=(\os_1, ..., \os_{k+1})\in\R^{(k+1)m}$, let $\mD_i = \omX_i\omS^{-1}_i$ for $i=1,...,k+1$.
Further, let $\mD_\Sigma = \sum_{i=1}^{k+1} \mD_i$ and define $\mE$ as in \eqref{eq:ipm:defineE}.
Then

\begin{align*}
(\cA^\top \omX \omS^{-1} \cA)^{-1}
= 
\left[\begin{array}{ccc|c}
\mI&&&0\\
&\ddots&&\vdots\\
&&\mI&0\\
\hline
-\mD_\Sigma^{-1}\mD_1\mA&\cdots&-\mD_\Sigma^{-1}\mD_k\mA&\mI
\end{array}\right]
\left[\begin{array}{ccc|c}
    &&&0\\
    &\mE^{-1}& &\vdots\\
    &&&0\\
    \hline
    0&\cdots&0&\mI
\end{array}\right]
\left[\begin{array}{ccc|c}
    \mI&&&-\mA^\top\mD_1\mD_\Sigma^{-1}\\
    &\ddots&&\vdots\\
    &&\mI&-\mA^\top\mD_k\mD_\Sigma^{-1}\\
    \hline
    0&\cdots&0&\mD_\Sigma^{-1}
\end{array}\right]
\end{align*}
\end{lemma}
\begin{proof}
Since $\omX\omS^{-1}>0$ and $\cA$ has linearly independent columns, $(\cA^\top \omX \omS^{-1} \cA)$ is positive definite.
By \Cref{lem:structure:hessian}, we have
    \begin{align*}
    \cA^\top \omX \omS^{-1} \cA
    = 
    \left[\begin{array}{ccc|c}
    \mA^\top \mD_1 \mA   &      &  0  &   \mA^\top \mD_1\\
     & \ddots &  &  \vdots \\
    0  &   &   \mA^\top\mD_k\mA  &   \mA^\top\mD_k\\
    \hline 
    \mD_1\mA  &  \cdots & \mD_k\mA   &    \mD_\Sigma
    \end{array}\right]=:
    \begin{bmatrix}A&B\\C&D\end{bmatrix}
    \end{align*}
We then apply \Cref{lem:structure:psd_inverse}. Since
\begin{align*}
E&=A-BD^{-1}C\\
&=
    \left[\begin{array}{ccc}
    \mA^\top \mD_1 \mA &  &   0  \\
     & \ddots & \\
    0   &  & \mA^\top\mD_k\mA 
    \end{array}\right]-
    \left[\begin{array}{c}
    \mA^\top \mD_1\\
    \vdots \\
    \mA^\top\mD_k\\
    \end{array}\right]
    \mD_\Sigma^{-1}
    \left[\begin{array}{ccc}
    \mD_1\mA  &  \cdots & \mD_k\mA
    \end{array}\right]%
\end{align*}
which is exactly the definition of $\mE$ in \eqref{eq:ipm:defineE},
substituting everything into \eqref{eq:structure:psd_inverse}
we have the desired result.
\end{proof}

The following \Cref{lem:structure:delta} tells us that \Cref{alg:commodity_ipm} and \Cref{alg:robust_ipm} perform indeed the same steps to $x$ and $s$.

\begin{lemma}\label{lem:structure:delta}
For given $\ox=(\ox_1, ..., \ox_{k+1}),\os=(\os_1, ..., \os_{k+1})\in\R^{(k+1)m}$, let $\mD_i = \omX_i\omS^{-1}_i$ for $i=1,...,k+1$.
Further, let $\mD_\Sigma = \sum_{i=1}^{k+1} \mD_i$ and define $\mE$ as in \eqref{eq:ipm:defineE}.

Let $w,v_i,\delta_s^i,\delta_x^i$ (for $i=1,...,k+1$) be defined as in \Cref{alg:commodity_ipm}.
Then
\begin{align*}
\delta_s =&~ \cA(\cA^\top\omX\omS^{-1}\cA)^{-1}\cA^\top\omS^{-1}\frac{t'}{32\lambda} \frac{g}{\|g\|_2}  \\
\delta_x =&~ \omS^{-1}\frac{t'}{32\lambda} \frac{g}{\|g\|_2} - \omX\omS^{-1}\delta_s \notag
\end{align*}
In particular, \Cref{alg:commodity_ipm} takes the same steps $\delta_x,\delta_s$ as \Cref{alg:robust_ipm}.
\end{lemma}

\begin{proof}
We prove $\frac{32\lambda\|g\|_2}{t'}\delta_s = \cA(\cA^\top\omX\omS^{-1}\cA)^{-1}\cA^\top\omS^{-1} g$ by rewriting the right-hand expression step by step.
We start with
$\cA^\top \omS^{-1} g$
which can be written as
$$
\cA^\top \omS^{-1} g
=
    \left[
    \begin{array}{ccc|c}
       \mA^\top & & 0 & 0 \\
        & \ddots & & \vdots \\
        0 & & \mA^\top & 0 \\
        \hline
        \mI & \cdots & \mI & \mI
    \end{array}
    \right]\omS^{-1}g
    =
    \begin{pmatrix}
        \mA^\top \omS_1^{-1}g_1 \\
        \vdots \\
        \mA^\top \omS_k^{-1} g_k \\
        \sum_{i=1}^{k+1} \omS_i^{-1} g_i
    \end{pmatrix}
$$
where the bottom block is exactly $\mD_\Sigma w$ from \Cref{alg:commodity_ipm}.

Next, we must multiply this vector by $(\cA^\top\omX\omS^{-1}\cA)^{-1}$ which by \Cref{lem:structure:inverse} can be written as
$$
(\cA^\top\omX\omS^{-1}\cA)^{-1}
=
\left[\begin{array}{ccc|c}
\mI&&&0\\
&\ddots&&\vdots\\
&&\mI&0\\
\hline
-\mD_\Sigma^{-1}\mD_1\mA&\cdots&-\mD_\Sigma^{-1}\mD_k\mA&\mI
\end{array}\right]
\left[\begin{array}{ccc|c}
    &&&0\\
    &\mE^{-1}& &\vdots\\
    &&&0\\
    \hline
    0&\cdots&0&\mI
\end{array}\right]
\left[\begin{array}{ccc|c}
    \mI&&&-\mA^\top\mD_1\mD_\Sigma^{-1}\\
    &\ddots&&\vdots\\
    &&\mI&-\mA^\top\mD_k\mD_\Sigma^{-1}\\
    \hline
    0&\cdots&0&\mD_\Sigma^{-1}
\end{array}\right]
$$
We analyze the product of this matrix with $\cA^\top\omS^{-1} g$ by multiplying from right to left. We start with
$$
\left[\begin{array}{ccc|c}
    \mI&&&-\mA^\top\mD_1\mD_\Sigma^{-1}\\
    &\ddots&&\vdots\\
    &&\mI&-\mA^\top\mD_k\mD_\Sigma^{-1}\\
    \hline
    0&\cdots&0&\mD_\Sigma^{-1}
\end{array}\right]
\begin{pmatrix}
        \mA^\top \omS_1^{-1}g_1 \\
        \vdots\\
        \mA^\top \omS_k^{-1} g_2 \\
        \sum_{i=1}^{k+1} \omS_i^{-1} g_i
    \end{pmatrix}
=
\begin{pmatrix}
        \mA^\top (\omS_1^{-1} g_1 - \mD_1 w) \\
        \vdots \\
        \mA^\top (\omS_k^{-1} g_k - \mD_k w) \\
        w
    \end{pmatrix}.
$$
Next we multiply
$$
\left[\begin{array}{ccc|c}
    &&&0\\
    &\mE^{-1}& &\vdots\\
    &&&0\\
    \hline
    0&\cdots&0&\mI
\end{array}\right]
\begin{pmatrix}
        \mA^\top (\omS_1^{-1} g_1 - \mD_1 w) \\
        \vdots\\
        \mA^\top (\omS_k^{-1} g_k - \mD_k w) \\
        w
    \end{pmatrix}
    =
    \begin{pmatrix}
        \mE^{-1} \begin{pmatrix}\mA^\top (\omS_1^{-1} g_1 - \mD_1 w) \\
        \vdots \\
        \mA^\top (\omS_k^{-1} g_k - \mD_k w)
        \end{pmatrix}\\
        w
    \end{pmatrix}
    =
    \begin{pmatrix}
        v_1 \\
        \vdots \\
        v_k \\
        w
    \end{pmatrix}.
$$
Here the vectors $v_1,...,v_k$ are exactly as defined in \Cref{alg:commodity_ipm}.
Finally, we can show
$$
(\cA^\top\omX\omS^{-1}\cA)^{-1} \cA^\top \omS^{-1} g
=
\left[\begin{array}{ccc|c}
\mI&&&0\\
&\ddots&&\vdots\\
&&\mI&0\\
\hline
-\mD_\Sigma^{-1}\mD_1\mA&\cdots&-\mD_\Sigma^{-1}\mD_k\mA&\mI
\end{array}\right]
\begin{pmatrix}
        v_1 \\
        \vdots \\
        v_k \\
        w
    \end{pmatrix}
    =
\begin{pmatrix}
        v_1 \\
        \vdots\\
        v_k \\
        w-\sum_{i=1}^k\mD_\Sigma^{-1}\mD_i\mA v_i
    \end{pmatrix}
$$
To prove $\frac{32\lambda\|g\|_2}{t'}\delta_s = \cA(\cA^\top\omX\omS^{-1}\cA)^{-1}\cA^\top\omS^{-1}g$
we are left with
\begin{align*}
\cA (\cA^\top\omX\omS^{-1}\cA)^{-1}\cA^\top\omS^{-1}g
=&~
    \left[\begin{array}{ccc|c}
       \mA & & 0 & \mI \\
        & \ddots & & \vdots \\
        0 & & \mA & \mI \\
        \hline
        0 & \cdots & 0 & \mI
    \end{array}\right]
    \begin{pmatrix}
        v_1 \\
        \vdots \\
        v_k \\
        w-\sum_{i=1}^k\mD_\Sigma^{-1}\mD_i\mA v_i
    \end{pmatrix} \\
=&~
\left(\begin{array}{lcl}
        w + \mA v_1 & - &\sum_{i=1}^k \mD_\Sigma^{-1}\mD_i\mA v_i \\
        &\vdots& \\
        w + \mA v_k & - &\sum_{i=1}^k \mD_\Sigma^{-1}\mD_i\mA v_i \\
        w & - &\sum_{i=1}^k \mD_\Sigma^{-1}\mD_i\mA v_i
\end{array}\right)\\
=&~
\begin{pmatrix}
    \delta_s^1 \\
    \vdots \\
    \delta_s^{k+1}
\end{pmatrix} \frac{32\lambda\|g\|_2}{t'}
\end{align*}

where $\delta_s^i$ for $i=1,...,k+1$ are as defined in \Cref{alg:commodity_ipm}.
In particular, this implies that $\delta_s$ are the same in \Cref{alg:commodity_ipm} and \Cref{alg:robust_ipm}.
The proof for $\delta_x$ follows directly from definition of $\delta_x$ in \Cref{alg:commodity_ipm}.
\end{proof}

\begin{proof}[Proof of \Cref{thm:commodity_ipm}]
    We argue that \Cref{alg:commodity_ipm} performs the same steps as \Cref{alg:robust_ipm}. \Cref{thm:commodity_ipm} thus follows from \Cref{lem:robust_ipm}.

    As can be seen in \Cref{alg:commodity_ipm} and \Cref{alg:robust_ipm}, both algorithms pick $\ox,\os,\ov,g,t'$ in the same way.
    Only the description of the steps $\delta_x,\delta_s$ differs.
    By \Cref{lem:structure:delta} these steps are the same vectors, so \Cref{alg:commodity_ipm} computes the same vectors $x+\delta_x$, $s+\delta_s$ as \Cref{alg:robust_ipm}.
\end{proof}

\subsection{Bounds on Primal and Dual}

We must bound the largest and smallest values that occur in $x$ and $s$ throughout \Cref{alg:commodity_ipm}.
This is because the complexity of the data structures in \Cref{sec:vector} scale in the log of the ratio of largest to smallest entry. So we must guarantee that these ratios are polynomially bounded.

\begin{lemma}%
\label{lem:ratiobound}
Let $(x, y, s)$ be a feasible point with $xs\approx_{1/10} t$ and $w' = (x', y', s')$ be a feasible point with $x's'\approx_{1/10} t'$ for $t' \le t$. 

Then $\|x'/x\|_1 + \|s'/s\|_1 \le 3(k+1)m$.
\end{lemma}

\cite[Lemma 3.3]{DadushHNV20} states \Cref{lem:ratiobound} in a slightly different form. They use the $\ell_\infty$-norm instead of the $\ell_1$-norm, and they use $xs=t$ instead of $xs\approx_{1/10}t$. However, their proof also directly implies \Cref{lem:ratiobound} as stated above. We repeat their proof here for completeness sake.
\begin{proof}[Proof of \Cref{lem:ratiobound}]
    By $\cA^\top (x-x') = 0$ and $s-s'=\cA(y-y')$ we have $(x-x')^\top(s-s') = 0$.
    This can be rewritten as 
    $$x^\top s' + x'^\top s = x^\top s+x'^\top s'.$$
    By assumption $xs\approx_{1/10} t$, $x's'\approx_{1/10} t'$, and $t' \le t$, the right hand side is upper bounded by $3(k+1)mt$.
    Dividing by $t$, and using that $t\le 2xs$ we obtain
    \begin{align*}
    \|\frac{x'}{x}\|_1 + \|\frac{s'}{s}\|_1
    =&~
    \sum_{i=1}^{(k+1)m} \left(\frac{x'_i}{x_i} + \frac{s'_i}{s_i}\right)
    =
    \sum_{i=1}^{k+1)m} \left(\frac{x'_i s_i}{x_i s_i} + \frac{s'_i x_i}{s_i x_i} \right)\\
    \le&~
    2
    \sum_{i=1}^{k+1)m} \left(\frac{x'_i s_i}{t} + \frac{s'_i x_i}{t}\right) \\
    =&~
    2 (x'^\top s + x^\top s') / t\\
    <&~
    3(k+1)m.
    \end{align*}
\end{proof}

\begin{corollary}\label{cor:bitlengthbound}
    Throughout the IPM we have $\frac{t}{10u} \le s \le 3(k+1)m\cdot s^\init$
    and $\frac{t}{(3(k+1)m s^\init}\le x \le u$
    where $u\in\R^{(k+1)m}$ is the vector of edge capacities repeated $k+1$ times.
\end{corollary}
\begin{proof}
    At the start of the IPM we have $x^\init s^\init \approx_{1/16} t^\init$
    and later during the IPM we always have $xs \approx_{1/16} t < t^\init$ (see \Cref{thm:commodity_ipm}).
    So by \Cref{lem:ratiobound} we have $s/s^\init \le 3(k+1)m$, so $s \le 3(k+1)m\cdot s^\init$.
    In the other direction, we have $xs\approx_{1/16} t$ and $x \le u$,
    so $s \ge 0.1t/u$.
    The lower bound on $x$ comes from $xs\approx_{1/16} t$.
\end{proof}

\section{Vector Maintenance}
\label{sec:vector}
In \Cref{sec:IPM} be described what the interior point method looks like when applied to a multi-commodity flow instance. 
We observed (see \Cref{line:commodityipm:deltaxdeltas} in \Cref{alg:commodity_ipm}) that the update to the primal solution $x$ and slack of the dual $s$ are given by the following sums \eqref{eq:vector:s},\eqref{eq:vector:x}.
Here we write $x^{(t)},s^{(t)}$ for $x,s$ constructed at the end of the $t\th$ iteration of the interior point method. 
Further, $\beta \in \R$ is the normalization $t'/(32\lambda\|g\|_2)$  
and $z_j = \omS_j^{-1}g_j+\mD_j w$ from \Cref{alg:commodity_ipm}. 

For $j=1,...,k$ we have in \Cref{alg:commodity_ipm} that $x,s$ are of the form (if we define $v_{k+1}=0$):
\begin{align}
    s^{(t)}_j 
    :=&~
    s^{(t-1)}_j + \left(
        w
        + \mA v_j
        - \sum_{\ell=1}^{k} \mD_\ell(\mD_{\Sigma})^{-1} \mA v_\ell
        \right) \beta^{(t)} \notag\\
    =&~
    s^{(t-1)}_j + \left(
        w
        + \sum_{\ell=1}^{k+1} \mD_\ell(\mD_{\Sigma})^{-1} \mA (v_j - v_\ell)
        \right) \beta^{(t)} \label{eq:vector:s}
    \\
    x^{(t)}_j 
    :=&~ 
    x^{(t-1)}_j + \left(
        z_j
        + \mD_j\left(\mA v_j
        - \sum_{\ell=1}^{k} \mD_\ell(\mD_{\Sigma})^{-1} \mA v_\ell
        \right)\right) \beta \notag\\
        =&~
    x^{(t-1)}_j + \left(
        z_j
        + \mD_j\left(
        \sum_{\ell=1}^{k+1} \mD_\ell(\mD_{\Sigma})^{-1} \mA (v_j - v_\ell)
        \right)\right) \beta^{(t)} \label{eq:vector:x}
\end{align}
(Note that by letting $v_{k+1}=0$ we can write $x_{k+1},s_{k+1}$ in \Cref{alg:commodity_ipm} in this form as well.)

The vectors $x^{(t)},s^{(t)}$ are $m$-dimensional vectors and we cannot afford to write them down in each iteration as that takes $\Omega(m)$ for a total of $\Omega(m^{1.5})$ over all iteration of the interior point method. Observe that \Cref{alg:commodity_ipm} does not actually need access to $x^{(t)},s^{(t)}$, but entry-wise approximations $\ox^{(t)},\os^{(t)}$ suffice, where \Cref{alg:commodity_ipm} then uses $d_j = \ox^{(t-1)}_j/\os^{(t-1)}_j$ for all $j=1,...,k+1$ and $d_\Sigma = \sum_{j=1}^{k+1} d_j$. 
So in this section we want to create/present a data structure that maintains these approximations.

\begin{theorem}\label{thm:primaldual:maintenance}
There exists a deterministic data structure with the following operations
\begin{itemize}
    \item \textsc{Initialize}$(\mA\in\R^{m\times n},\epsilon\in(0,1], w\in\R^m, (z_j,s^{(0)}_j,x^{(0)}_j)_{1\le j \le k+1} \in \R_{>0}^m \times \R^m \times \R_{>0}^m \times \R_{>0}^m)$\\
    Initialize on the given incidence matrix $\mA$, edge weights $z_1,\ldots,z_{k+1},w,s^{(0)}_1,...,s^{(0)}_{k+1}$ and accuracy-parameter $\epsilon$ in $\hO(k^2m)$ time.
    \item \textsc{Update}$(i\in [m], c \in \R_{>0}^{k+1},c'\in\R)$ Set $(z_j)_i\leftarrow c_j$ for all $j\in[k]$, $w_i \leftarrow c'$, 
    in $\hO(k)$ amortized time.
    \item \textsc{Add}$(v_1,\ldots,v_k\in\R^n,\beta>0)$
    Let $v_{k+1}=0$. Let $\ox_j^{(t-1)},\os_j^{(t-1)}$ be the output from the last call to \textsc{Add}.
    Let $d_i = \ox_i^{(t-1)}/\os_i^{(t-1)}$ and $d_\Sigma := \sum_{j=1}^{k+1} d_j$ and let $x^{(t)}_j,s^{(t)}_j$ for $j=1,...,k$ as in \eqref{eq:vector:s},\eqref{eq:vector:x}.
    The $t\th$ call to \textsc{Add} returns $k+1$ vectors $\os^{(t)}_j,\ox^{(t)}_j$ for $j=1,...,k+1$ with
    $$
    \os_j^{(t)} \approx_\epsilon s_j^{(t)} \text{ for all }  j\in[k+1]
    $$
    $$
    \ox_j^{(t)} \approx_\epsilon x_j^{(t)} \text{ for all }  j\in[k+1]
    $$
    The vectors $\os^{(t)}_j,\ox^{(t)}_j$ for $j=1,...,k+1$ are returned as a pointer, together with a list $I\subset [k+1]\times[m]$ of indices $(j,i)$ where $(\ox_j^{(t)})_i$ or $(\os_j^{(t)})_i$ changed compared to $(\ox_j^{(t-1)})_i,(\os_j^{(t-1)})_i$.
    \item \textsc{Exact}$()$ Returns $s^{(t)}_j$ and $x^{(t)}_j$ for all $j=1,...,k+1$ in $O(k^2m)$ time.
\end{itemize}
The total time of the first $T\le \sqrt{m}$ calls to \textsc{Add} is bounded by $\tO(k^2 m^{o(1)} + Tkn\log W)$, if for all $t=1,...,T$ there is some $\mu^{(t)}\in\R_{>0}$ with
\begin{align*}
\sum_{j=1}^{k+1}\|(s_j^{(t-1)})^{-1} (s^{(t)}_j - s^{(t-1)}_j)\|_2^2 \le 1/10^2,
~~~
\sum_{j=1}^{k+1}\|(x_j^{(t-1)})^{-1} (x^{(t)}_j - x^{(t-1)}_j)\|_2^2 \le 1/10^2,
~~~ x^{\ell} s^{\ell} \approx_{1/10} \mu^{(t)}
\end{align*}
and $W$ upper bounds the ratio of largest to smallest entry of any $x_j^{(t)}$ or $s_j^{(t)}$.
\end{theorem}

Before we prove \Cref{thm:primaldual:maintenance}, we first outline its proof and state some useful lemmas.
For simplicity, we will focus in our outline of \Cref{thm:primaldual:maintenance} only on the slack of the dual $s^{(t)}$.

In \cite{BrandLSS20,BrandLN+20,BrandLL+21}, data structures were given that maintain an approximation 
$\os'^{(t)} \approx s'^{(t)} = s'^{(t-1)} + \mA h$.
These can easily be extended to maintain an approximation of 
$s'^{(t)} = s'^{(t-1)} + (w + \mG\mA h) \beta$
for some diagonal matrix $\mG$, vector $w$, and scalar $\beta$.
Such a data structure is given by \Cref{lem:vector:maintenance} and proven in \Cref{sec:vector:maintain}.
If we run $k$ copies of this data structure, we can maintain an approximation of each term of the sum in \eqref{eq:vector:s} by letting $\mG = \mD_\ell\mD_\Sigma^{-1}$ and $h = v_j - v_\ell$.

\begin{restatable}{lemma}{vectormaintenance}\label{lem:vector:maintenance}
There exists a deterministic data structure with the following operations
\begin{itemize}
    \item \textsc{Initialize}$(\mA\in\R^{m\times n},\epsilon\in(0,1]^m, g\in\R^m_{>0}, w\in\R^m, s^{(0)}\in\R^m)$
    Initialize on the given incidence matrix $\mA$, edge weights $g$ and accuracy-vector $\epsilon$ in $\hO(m)$ time.
    \item \textsc{SetAccuracy}$(i\in [m], \delta\in(0,1])$ Set $\epsilon_i \leftarrow \delta$ in $\hO(1)$ amortized time.
    \item \textsc{Update}$(i\in [m], c \in \R_{>0},c'\in\R)$ Set $g_i\leftarrow c$, $w_i\leftarrow c'$ in $\hO(1)$ amortized time.
    \item \textsc{Add}$(h\in\R^n,\beta>0)$
    Let $h^{(\ell)}$ be the vectors $h$ given during the $\ell\th$ call to \textsc{Add}.
    Let $g^{(\ell)}$ be the state of $g$ during the $\ell\th$ call to \textsc{Add}.
    Let 
    $$
    s^{(t)} := s^{(0)} + \left(\sum_{\ell=1}^t \mG^{(\ell)} \mA h^{(\ell)} + \beta^{(\ell)}w^{(\ell)}\right)
    $$
    The $t\th$ call to \textsc{Add} returns a vector $\os$ with
    $$
    |\os_i - s^{(t)}_i| \le \epsilon_i \text{ for all }  i
    $$
    The vector $\os$ is returned as a pointer, together with a list $I\subset [m]$ of indices where $\os_i$ changed.
    \item \textsc{Exact}$(i\in [m])$ Returns $s^{(t)}_i$ in $O(1)$ time where $t$ is the number of calls to \textsc{Add} so far.
\end{itemize}
The total time of the first $T\le\sqrt{m}$ calls to \textsc{Add} is bounded by 
$$\tO(m^{1+o(1)} + m^{o(1)}T(\sum_{\ell=1}^{T}\|(\epsilon^{(\ell)})^{-1} \mG^{(\ell)}\mA h^{(\ell)}\|_2^2+\|(\epsilon^{(\ell)})^{-1} w^{(\ell)}\beta^{(\ell)}\|_2^2) + Tn\log W))$$
where $W$ bounds the largest ratio of largest to smallest entry of $(\epsilon^{(\ell)})^{-1}\mG^{(\ell)}$ for all $\ell=1,...,T$.
\end{restatable}

There is one main issue with this approach of running several copies of \Cref{lem:vector:maintenance} to maintain $\ox^{(t)},\os^{(t)}$ as in \Cref{thm:primaldual:maintenance}.
Notice that the complexity of \Cref{lem:vector:maintenance} depends on
$\|(\epsilon)^{-1} \mG\mA h\|_2^2$ and $\|(\epsilon)^{-1} w \beta\|_2^2$
which for our application will be
$\epsilon = O(1/\os_j)$ (since we want a multiplicative approximation instead of an additive one)
and $\mG = \mD_\ell\mD_\Sigma^{-1}$, and $h = (v_j - v_\ell)\beta$.
So the complexity of running all these copies in parallel would depend on
\begin{align}
\sum_{j=1}^{k+1} \left(\|s_j^{-1}w\beta\|_2^2  + \sum_{\ell=1}^k \|s_j^{-1}(\mD_\ell\mD_\Sigma^{-1}\mA(v_j - v_\ell)\beta)\|_2^2\right) \label{eq:vector:sumofnorms}
\end{align}
In general, these terms can be much larger than
\begin{align}
\sum_{j=1}^{k+1} \left\|s_j^{-1}\beta\left(w + \sum_{\ell=1}^k \mD_\ell\mD_\Sigma^{-1}\mA(v_j - v_\ell)\right)\right\|_2^2
=
\sum_{j=1}^{k+1} \| (s_j^{(t-1)})^{-1}(s_j^{(t)} - s_j^{(t-1)})\|_2^2 \label{eq:vector:normofsums}
\end{align}
which is what the complexity stated in \Cref{thm:primaldual:maintenance} is supposed to depend on.
It is important that we have a complexity dependence on \eqref{eq:vector:normofsums}
because a bound on this norm is given by \Cref{thm:commodity_ipm} (last bullet).

The following \Cref{lem:ipm:normbounds} shows that we can actually bound \eqref{eq:vector:sumofnorms} by \eqref{eq:vector:normofsums}.
\begin{restatable}{lemma}{normbounds}\label{lem:ipm:normbounds}
    For any $v_1, ...,v_k \in \R^n$, $v_k=0\in\R^n$, $w, \ox_1,...,\ox_{k+1},\os_1,...,\os_{k+1} \in \R_{>0}^m$ 
    and $d_i=\ox_i/\os_i$ for $i=1,2,\ldots,k+1$, 
    let $d_\Sigma=\sum_{i=1}^{k+1}d_i$,
    and let $\mu\in\R_{>0}$ such that $\ox_i\os_i \approx_{1/5} \mu$ for all $i=1,...,k+1$.

Suppose for some $\epsilon\ge 0$,
\begin{align*}
\sum_{i=1}^{k+1}\left\lVert\frac{1}{\os_1}\left(w+\sum_{j=1}^{k+1} \dfrac{d_j}{d_\Sigma}\mA (v_i-v_j)\right)\right\rVert_2^2\le\epsilon^2
\end{align*}

Then,
\begin{align*}
\sum_{i=1}^{k+1}\sum_{j=1}^{k+1}
\left\lVert\frac{1}{\os_i}\frac{d_j}{d}\mA(v_i-v_j)\right\rVert_2^2\le 6\epsilon^2\\
\sum_{i=1}^{k+1}\left\lVert\frac{1}{\os_i}w\right\rVert_2^2\le 2\epsilon^2
\end{align*}
\end{restatable}

\Cref{lem:ipm:normbounds} allows us to bound the complexity of running the copies of \Cref{lem:vector:maintenance} in parallel,
which then implies the time complexities as stated in \Cref{thm:primaldual:maintenance}.
We now prove \Cref{thm:primaldual:maintenance} using \Cref{lem:ipm:normbounds,lem:vector:maintenance}. 
The proof of \Cref{lem:ipm:normbounds} is deferred to the next subsection. \Cref{lem:vector:maintenance} is deferred to the appendix because it is a simple modification of data structures in \cite{BrandLSS20,BrandLN+20,BrandLL+21}.

\begin{algorithm2e}[ht!]
\caption{Vector Maintenance (\Cref{thm:primaldual:maintenance}) \label{alg:primaldual:maintenance}}
\SetKwProg{Proc}{procedure}{}{}
\SetKwProg{Params}{parameters}{}{}
\Proc{\textsc{Initialize}$(\mA,\epsilon,w,(z_j,s_j^{(0)},x_j^{(0)})_{1\le j\le k+1})$}{
    $\os_j \leftarrow s^{(0)}_j$, $\ox_j \leftarrow x^{(0)}_j$, $d_j \leftarrow \ox_j/\os_j$ for $j=1,...,k+1$. \\
    $d_\Sigma \leftarrow \sum_{j=1}^{k+1} d_j$ \\
    \tcp{Data structures used to maintain $\os_j \approx s_j$}
    $D^s_{j,\ell}.\textsc{Initialize}(\mA,\epsilon/(10\os_j), d_\ell/d_\Sigma, 0, 0)$ for $j=1,...,k+1,\ell=1,...,k$ \\
    $D^s_{j,k+1}.\textsc{Initialize}(\mA,\epsilon/(10\os_j), d_{k+1}/d_\Sigma, w, 0)$ for $j=1,...,k+1$ \\
    \tcp{Data structures used to maintain $\ox_j \approx x_j$}
    $D^x_{j,\ell}.\textsc{Initialize}(\mA,\epsilon/(10\ox_j), d_j d_\ell/d_\Sigma, 0, 0)$ for $j=1,...,k+1,\ell=1,...,k$ \\
    $D^x_{j,k+1}.\textsc{Initialize}(\mA,\epsilon/(10\ox_j), d_j d_{k+1}/d_\Sigma, z_j, 0)$ for $j=1,...,k+1$
}
\Proc{\textsc{Update}$(i\in[m],c,c')$}{
    \tcp{Update data structure to use the new $w_i = c$}
    $D^s_{j,k+1}.\textsc{Update}(i, c')$ for $j=1,...,k+1$ \\
    \tcp{Update data structure to use the new $(z_j)_i = c_j$ for $j=1,...,k+1$}
    $D^x_{j,k+1}.\textsc{Update}(i, c_j)$ for $j=1,...,k+1$
}
\Proc{$\textsc{Add}(v_1,...,v_k\in\R^n, \beta>0)$}{
	$t \leftarrow t+1$, let $v_{k+1}=0$ \\
    \tcp{For simplicity we describe the procedure only for $\os$, $\ox$ works similarly.}
    \For{$j=1,...,k+1$}{
        \tcp{$u_j$ is approximation of $s^{(t)}_j := s^{(t-1)}_j + (w - \sum_{\ell=1}^{k+1} \mD_\ell\mD_\Sigma^{-1}\mA (v_j - v_\ell))\beta$}
        $u_{j}, I_j \leftarrow \sum_{\ell=1}^{k+1} D^s_{j,\ell}.\textsc{Add}(v_j - v_\ell, \beta)$ \tcp*[h]{$w$ is contained in $D^s_{j,k+1}$}\label{line:primaldual:u}\\
        For indices $i \in I_k$ set $(\os_j)_i \leftarrow (u_j)_i$ if $|(\os_j)_i - (u_j)_i|>\epsilon/(5\os_i)$ \label{line:primadual:updateos}\\
        Let $J_j\subset I_j$ be the indices where we changed $(\os_j)_i$.\\
        \For{$i\in J_j$ \label{line:primadual:forupdatedos}}{
            \tcp{Update data structure accuracy so the additive error becomes multiplicative error}
            $D^s_{j,\ell}.\textsc{SetAccuracy}(i,\epsilon/(10(\os_j)_i))$ for $\ell=1,...,k+1$.\\
            \tcp{Update data structure so they use the new $\mD_\ell\mD_\Sigma^{-1}$}
            $(d_j)_i \leftarrow (\ox_j)_i/(\os_j)_i$ and update $d_\Sigma = \sum_{\ell=1}^{k+1} d_j$\\
            $D^s_{j',\ell}.\textsc{Update}(i,(d_\ell/d_\Sigma)_i)$ for $\ell,j'=1,...,{k+1}$.\\
        }
    }
	\Return $\os_j$, $J_j$ for $j=1,...,{k+1}$.
}
\end{algorithm2e}

\begin{proof}[Proof of \Cref{thm:primaldual:maintenance}]
    The algorithm description is given in \Cref{alg:primaldual:maintenance}. 
    As outlined in this section, the idea is to run $O(k^2)$ copies of \Cref{lem:vector:maintenance} to approximate the vectors $s_j^{(t)}$ for $j=1,...,k+1$.
    For this, we run a copy of \Cref{lem:vector:maintenance} for $\mG = \mD_\ell\mD_\Sigma^{-1}$, $h=(v_j-v_\ell)\beta$ for each $j,\ell=1,...,k+1$. All copies except for those with $\ell=k+1$ will have $w=0$.
    Thus the sums $u_j$ (\Cref{line:primaldual:u} in \Cref{alg:primaldual:maintenance}) of the vectors maintained by \Cref{lem:vector:maintenance} will approximate the vectors $s_j^{(t)}$ defined in \eqref{eq:vector:s}.

    Our proof focuses on how to maintain $s^{(t)}_j$ as in \eqref{eq:vector:s}. 
    The proof for $x^{(t)}_j$ is the same, we just change $\mG = \mD_j\mD_\ell\mD_\Sigma{-1}$ for $j,\ell=1,...,k+1$ and replace $w$ by $z_j$.

    \paragraph{Correctness}
    We always have that $|(u_j)_i - (s_j)_i| < \epsilon/(10(\os_j)_i)$ for all $j\in[k],i\in[m]$ by \Cref{lem:vector:maintenance} and our choice to use $\epsilon/(10k\os_j)$ as accuracy parameter for \Cref{lem:vector:maintenance}.

    By \Cref{line:primadual:updateos} we thus have $|(\os_j)_i - (s_j)_i| \le |(\os_j)_i - (u_j)_i| + |(s_j)_i - (u_j)_i| < \epsilon/(4(\os_j)_i)$. So $\os_j \approx_\epsilon s_j$ is a valid approximation.
    
    \paragraph{Complexity}
    The $\hO(k^2m)$ complexity of \textsc{Initialize} 
    come from the fact that we initialize $O(k^2)$ instances of \Cref{lem:vector:maintenance}
    which take $\hO(m)$ time each.

    The $\hO(k)$ complexity of \textsc{Update} comes from perform $O(k)$ calls to \textsc{Update} of \Cref{lem:vector:maintenance} which takes $\hO(1)$ time each.

    We now bound the total time of $T$ calls to \textsc{Add}.
    The complexity cost of the calls to \textsc{Add} of \Cref{lem:vector:maintenance} can be bounded by
    $$
    \tilde{O}\left(
    k^2Tn\log W + m^{o(1)}\sum_{t=1}^T\sum_{j=1}^k \|(s_j^{(t-1)})^{-1} (s_j^{(t)} - s_j^{(t-1)})\|_2^2/\epsilon^2
    \right)
    =
    \tilde{O}(k^2m^{1+o(1)}/\epsilon^2 + k^2Tn\log W)
    $$
    by using the norm bound from \Cref{lem:ipm:normbounds} and the assumption $\sum_{j=1}^{k+1} \|(s_j^{(t-1)})^{-1} (s_j^{(t)} - s_j^{(t-1)})\|_2^2 \le O(1)$ and $T \le \sqrt{m}$.

    The cost of \textsc{SetAccuracy} and \textsc{Update} performed in the loop of \Cref{line:primadual:forupdatedos} is $\hO(k^2)$ per entry $i\in J_j$
    because we must update all $O(k^2)$ copies of \Cref{lem:vector:maintenance}.
    For an index $i$ to be in $J_j$, the entry $(s_j)_i$ must have changed by some $\Omega(\epsilon s_j)$ as otherwise there was no need to update $(\os_j)_i$.
    By the norm bound $\|(s^{(t-1)}_j)^{-1}(s^{(t)}_j-s^{(t-1)}_j)\|_2 \le 1/10$ we can have at most $O(T^2/\epsilon^2) = O(m/\epsilon^2)$ such changes over $T \le \sqrt{m}$ iterations.

    In summary, the total cost of all $\sqrt{m}$ calls to \textsc{Add} is bounded by
    $$
    \tilde{O}(
    k^2m^{1+o(1)} + Tkn\log W
    ).
    $$
\end{proof}

The following lemma allows for a better amortized complexity when using the output $\ox_j,\os_j$ of our data structure as input to an algebraic data structure that uses fast matrix multiplication.
\Cref{lem:vector:amortizednumchanges} essentially states that large changes to $\ox_j,\os_j$ happen infrequently.

\begin{lemma}\label{lem:vector:amortizednumchanges}
We can assume that the output $\ox_j,\os_j$ for $j=1,...,k+1$ of \Cref{thm:primaldual:maintenance} satisfies the following:
Every $2^i$ calls to \textsc{Add}, at most $\tilde{O}(2^i/\epsilon^2)$ entries of $\ox_j,\os_j$ change in total over all $j=1,...,k+1$.
\end{lemma}

It was independently proven \cite{CohenLS19,Brand20,BrandLSS20} (with slightly different argument) that this assumption can be made on the approximate $\ox,\os$.
In previous work, this was stated for a specific sequence coming from certain data structures. Here we state it in general form via \Cref{lem:vector:stabilizer}.
Using the following \Cref{lem:vector:stabilizer} to update $\os_j$ in \Cref{line:primadual:updateos}
implies \Cref{lem:vector:amortizednumchanges} (here we use \Cref{lem:vector:stabilizer} for $\ov = \ln u_j$, $v=\ln s$ and $\os = \exp(\ov')$ and increase the accuracy of $u_j$ as maintained in \Cref{alg:primaldual:maintenance} by some $O(1/\log(m))$ factor).

\begin{lemma}[{Modification of \cite[Lemma 19]{LeeV21}}]\label{lem:vector:stabilizer}
    Assume we are an online sequence of vectors $\ov^{(0)},\ov^{(1)},...\in\R^m$ arriving in a stream, implicitly given via their $\Delta^{(t)} = \ov^{(t)}- \ov^{(t-1)}$. 
    Assume further that there exists another sequence $v^{(1)},v^{(1)},...$ with $\|v^{(t)}-v^{(t-1)}\|_2 \le \beta$ and $\|v^{(t)} - \ov^{(t)}\|_\infty \le \beta/(16\log m)$ for all $t$.

    Then we can construct a sequence
    $\ov'^{(1)},\ov^{(2)},...$
    in amortized time $O(\nnz(\Delta^{(t)}) \log m)$ for iteration $t$,
    such that
    \begin{itemize}
        \item $\|\ov'^{(t)} - v^{(t)}\|_\infty \le \beta$ for all $t$
        \item for all $i$, every $2^i$ iterations, at most $O(2^{2i}\alpha^2\beta^{-1}\log^2 m)$ entries of the returned vector $\ov'^{(t)}$ change.
    \end{itemize}
\end{lemma}

\begin{proof}
    The algorithm is as follows:
    At the start, set $\ov'^{(0)} = \ov^{(0)}$.
    Then for all $0\le\ell\le \log\sqrt{m}$, every $2^\ell$ iterations, 
    let $I$ be the set of indices $i\in[m]$ where 
    $|\ov^{(t)}_i - \ov^{(t-2^i)}_i| \ge \beta/(4\log m)$.
    Set $\ov'^{(t)}_i = \ov^{(t)}_i$ for $j\in I$.
    After $t=\sqrt{m}$ iterations, we set $\ov'^{(t)} = \ov^{(t)}$ and restart.
    
    This can be done efficiently by keeping $\log m$ ordered lists of changed entries. The $\ell$th list keeps the list of the entries that changed over the past $2^\ell$ iterations. The entries are kept in order from highest to lowest. This can be done in $O(1)$ amortized time per changed entry and list, via a balanced binary search tree data structure.

    \paragraph{Approximation guarantee:}
    We argue that for all $t$ we have $\|\ov'^{(t)} - v^{(t)}\|_\infty \le \beta$.
    For any $i$, let $t'$ be the last time we set $\ov'^{(t')}_i = \ov^{(t')}_i$.
    If $t'=t$, then $|\ov'^{(t)}_i - v^{(t)}_i| = |\ov^{(t)}_i - v^{(t)}_i| \le \beta$ and we are done.
    Otherwise, there is a sequence of at most $2\log m$ many $t'=t_0<t_1...<t_k=t$
    such that each $t_j-t_{j-1}$ is a power of $2$.
    Here we have
    $$
    |\ov'^{(t)}_i - v^{(t)}_i|
    \le
    |\ov'^{(t)}_i - \ov^{(t)}_i| + |v^{(t)}_i - \ov^{(t)}_i|
    \le
    (\sum_{j=1}^k |\ov^{(t_{j-1})}_i - \ov^{(t_j)}_i|) + |v^{(t)}_i - \ov^{(t)}_i|
    \le
    \beta
    $$

    \paragraph{Number of changes}
    Consider an update when the number of updates is a multiple of $2^\ell$ for some $\ell$.
    For every $i\in I$ we have that $|\ov^{(t)}_i - \ov^{(t-2^\ell)}_i| \ge \beta/(4\log m)$ but that also means
    $|v^{(t)}_i - v^{(t-2^i)}_i| \ge \beta/(8\log m)$.
    Since we have $\|v^{(t)}-v^{(t-1)}\|_2 \le \alpha$ for all $t$, there can be at most $O(2^{2\ell}\alpha^2\beta^{-2}\log^2 m)$ such entries.
\end{proof}

\subsection{Norm Bounds}

The only remaining part for proving \Cref{thm:primaldual:maintenance} is to prove \Cref{lem:ipm:normbounds}.
We start with proving a small variation of it that was already presented in \Cref{sec:overview:singlecommodity} as \Cref{lem:overview:normbounds} for the special case $w=0$.
\begin{lemma}\label{lem:vector:simplifiednorm}
        For any $v_1, ...,v_{k+1} \in \R^n$, $w\in\R^m$
    $d_1,...,d_{k+1}\in\R^{m}_{>0}$ 
    let $d_\Sigma=\sum_{i=1}^{k+1}d_i$. Write $\mD_i$ and $\mD_\Sigma$ for the diagonal matrices with $d_i$ and $d_\Sigma$ on the diagonal.
    Then
    \begin{align*}
\sum_{i=1}^{k+1}\sum_{j=1}^{k+1}
\left\lVert\mD_i^{1/2}\mD_j\mD_\Sigma^{-1}\mA(v_i-v_j)\right\rVert_2^2
\le 4\cdot
\sum_{i=1}^{k+1} \left\|\mD^{1/2}_i \left(w+ \sum_{j=1}^{k+1} \mD_j\mD_\Sigma\mA(v_i - v_j)\right)\right\|_2^2
\end{align*}
\end{lemma}

\begin{proof}
For $i=1,\ldots,k+1$, let $\eta_i=w+\mA v_i-\sum_{j=1}^{k+1}\mD_j\mD^{-1}_\Sigma\mA v_j=w+\sum_{j=1}^{k+1} \mD_j\mD^{-1}_\Sigma \mA (v_i-v_j)$. Then,
\begin{align*}
\sum_{i=1}^{k+1}\sum_{j=1}^{k+1}
\left\lVert \mD_i^{1/2} \mD_j\mD^{-1}_\Sigma\mA(v_i-v_j)\right\rVert_2^2
&=
\sum_{i=1}^{k+1}\sum_{j=1}^{k+1}
\left\lVert \mD_i\mD_j^2\mD^{-2}_\Sigma(\eta_i-\eta_j)^2\right\rVert_1 \\
&\le 2\sum_{i=1}^{k+1}\sum_{j=1}^{k+1}
\left\lVert \mD_i\mD_j^2\mD^{-2}_\Sigma\eta_i^2\right\rVert_1+2\sum_{i=1}^{k+1}\sum_{j=1}^{k+1}
\left\lVert \mD_i\mD_j^2\mD^{-2}_\Sigma\eta_j^2\right\rVert_1\\
&\le 2\sum_{i=1}^{k+1}
\left\lVert \mD_i\eta_i^2\right\rVert_1+2\sum_{j=1}^{k+1}
\left\lVert \mD_j^2\mD^{-1}_\Sigma\eta_j^2\right\rVert_1\\
&\le 2\sum_{i=1}^{k+1}
\left\lVert \mD_i \eta_i^2\right\rVert_1+2\sum_{j=1}^{k+1}
\left\lVert \mD_j\eta_j^2\right\rVert_1\\
&\le 4\sum_{i=1}^{k+1}
\left\lVert \mD_i\eta_i^2\right\rVert_1\\
&=
4\sum_{i=1}^{k+1} \left\|\mD^{1/2}_i \left(w+\sum_{j=1}^{k+1} \mD_j\mD_\Sigma\mA(v_i - v_j)\right)\right\|_2^2
\end{align*}
\end{proof}
We now prove \Cref{lem:ipm:normbounds}.
\normbounds*

\begin{proof}
We have $\ox_i\os_i\approx_{1/5}\mu$ for all $i\in\{1,2,\ldots,k+1\}$.
Thus, 
$$\dfrac{1}{\os_i^2}\approx_{1/5}\frac{\ox_i\os_i}{\mu}~\frac{1}{\os_i^2} = \dfrac{d_i}{\mu}.$$

For $i=1,\ldots,k+1$, let $\eta_i=w+\mA v_i-\sum_{j=1}^{k+1}\dfrac{d_j}{d}\mA v_j=w+\sum_{j=1}^{k+1} \frac{d_j}{d} \mA (v_i-v_j)$. Then,
\begin{align*}
\frac{1}{\mu}\sum_{i=1}^{k+1}\lVert d_i \eta_i^2\rVert_1\approx_{1/5}\sum_{i=1}^{k+1}\lVert \frac{1}{\os_i^2} \eta_i^2\rVert_1=
\sum_{i=1}^{k+1}\lVert\frac{1}{\os_i}\eta_i\rVert_2^2\le\epsilon^2
\end{align*}
By \Cref{lem:vector:simplifiednorm} we have
\begin{align*}
\sum_{i=1}^{k+1}\sum_{j=1}^{k+1}
\left\lVert\frac{1}{\os_i}\frac{d_j}{d}\mA(v_i-v_j)\right\rVert_2^2
&\approx_{1/5}\sum_{i=1}^{k+1}\sum_{j=1}^{k+1}
\frac{1}{\mu}\left\lVert\sqrt{d_i}\frac{d_j}{d}\mA(v_i-v_j)\right\rVert_2^2\\
&\le \frac{4}{\mu}\sum_{i=1}^{k+1}
\left\lVert\sqrt{d_i}\eta_i\right\rVert_2^2\\
\sum_{i=1}^{k+1}\sum_{j=1}^{k+1}
\left\lVert\frac{1}{\os_i}\frac{d_j}{d}\mA(v_i-v_j)\right\rVert_2^2&\le e^{2/5}4\epsilon^2\le 6\epsilon^2
\end{align*}

Since $w=\sum_{i\in[k+1]}\dfrac{d_i}{d}\eta_i$,

\begin{align*}
\sum_{i\in[k+1]}\left\lVert\frac{1}{\os_i}w\right\rVert_2^2&\approx_{1/5}\sum_{i\in[k+1]}\lVert \frac{d_i}{\mu}w^2\rVert_1\\
&=\lVert \frac{d}{\mu}w^2\rVert_1\\
&=\lVert \frac{d}{\mu}(\sum_{i\in[k+1]}\frac{d_i}{d}\eta_i)^2\rVert_1\\
&\le \lVert \frac{d}{\mu}\sum_{i\in[k+1]}\frac{d_i}{d}\eta_i^2\rVert_1&&\text{by Jensen's inequality}\\
&\le \lVert \sum_{i\in[k+1]}\frac{d_i}{\mu}\eta_i^2\rVert_1\\
\sum_{i\in[k+1]}\left\lVert\frac{1}{\os_i}w\right\rVert_2^2&\le e^{2/5}\epsilon^2\le 2\epsilon^2
\end{align*}

\end{proof}

\section{Multi-Commodity Flow Algorithm}
\label{sec:algorithm}

In this section we prove our main results \Cref{thm:intro:mincost,thm:intro:throughput} -- solving multi-commodity flow in $\tilde{O}(k^{2.5}\sqrt{m}n^{\omega-1/2})$ time.
We restate one of the results here as a reminder.

\mincost*

We prove this by combining the data structures from \Cref{sec:vector} with the interior point method from \Cref{sec:IPM}, i.e.~we show that the interior point method can be implemented efficiently by using the data structures.

We start by handling some particularities of $k$-commodity flow: The matrix $\mB^\top\mB$ is not invertible if $\mB$ is an incidence matrix, so technically we can not directly apply the interior point method from \Cref{sec:IPM}. 
In \Cref{sec:algorithm:initial} we describe how to modify the $k$-commodity flow LP such that this issue can be resolved.
We also describe how to modify the LP such that we can find an initial point for the interior point method.

Then in \Cref{sec:algorithm:implement} we combine all our results to prove \Cref{thm:intro:mincost,thm:intro:throughput}

\subsection{Initial Point and Invertibility}
\label{sec:algorithm:initial}
Let us quickly recap the structure of the given linear program to see why the constraint matrix is not full-rank. Then we describe how to make the linear program full rank.

We are given a $k$-commodity flow instance on graph $G=(V,E)$ with two demands $d_1,...,d_k \in \R^V$, costs $c_1,...,c_k \in \R^E$, edge capacities $u \in \R^E_{>0}$. We can write this instance as an LP as follows.
Let $\mB\in\{-1,0,+1\}^{E\times V}$ be the edge-vertex-incidence matrix, then define
\begin{align}
    \cB^ := \begin{bmatrix}
        \mB & 0 & 0 & \mI \\
        0 & \ddots & 0 & \mI \\
        0 & 0 & \mB & \mI \\
        0 & ... & 0 & \mI
    \end{bmatrix} \in \R^{(k+1)E\times(kV+E)} \quad
    d = \begin{pmatrix}
        d_1 \\
        \vdots \\
        d_k \\
        u
    \end{pmatrix} \in \R^{kV+E} \quad
    c = \begin{pmatrix}
        c_1 \\
        \vdots \\
        c_k \\
        0
    \end{pmatrix} \in \R^{(k+1)E}
\end{align}
Then the $k$-commodity flow problem can be written as following primal and dual LP
\begin{align}
    (P)~\min_{\cB^\top x = d, x\ge 0} c^\top x  \quad\quad\quad (D)~\max_{\cB y + s = c, s\ge 0} b^\top y \label{eq:structure:lp}
\end{align}
Here $x$ can be split into $k+1$ many $m$-dimensional vectors $x = (x_1, ..., x_{k+1})$ where $x_1,..., x_k$ are the flows corresponding to the $k$ commodities and $x_{k+1} = u-\sum_{i=1}^k x_i$ is the slack for the capacities.
Note that the matrix $\cB$ is not full rank because $\mB$ is not full rank (e.g.~the vector $(\mathbf{1}_m,0,...,0) \in \R^{kE+V}$ is in the kernel of $\cB$).
To make the matrix full rank, we can modify the matrices as follows.

Assume for simplicity that the graph has only one weakly connected component (otherwise we can just solve the flow problem on each component independently).
Let $\mA$ be the edge-vertex incidence matrix $\mB$ with the first row deleted, and let $b_1,..,b_k$ be the vectors $d_1,...,d_k$ where the first entry was deleted. Further let
\[\cA := \begin{bmatrix}
        \mA & 0 & 0 & \mI \\
        0 & \ddots & 0 & \mI \\
        0 & 0 & \mA & \mI \\
        0 & ... & 0 & \mI
    \end{bmatrix} %
    \quad
    b = \begin{pmatrix}
        b_1 \\
        \vdots \\
        b_k \\
        u
    \end{pmatrix} %
\]
Then the following \Cref{lem:invertible} shows that the matrix becomes full-rank and that deleting the row did not change the solution of the linear program.
\begin{lemma}\label{lem:invertible}
If $\mathbf{1}^\top d = 0$ and the underlying graph is weakly connected, 
then for $x\in\R^{(k+1)E}$, $\cA^\top x=b$ iff $\cB x=d$.

Furthermore, the columns of $\cA$ are linearly independent.
\end{lemma}
\begin{proof}
The linear system $\cB x = d$ has linear dependent constraints as $\cB$ has rank $km+kn-k$ but consists of $km+kn$ columns. By deleting $k$ columns of $\cB$, we pick a maximum set of linear independent constraints. So the solution set stays the same.
The linear independence is given by $\cB^\top \cB$ being of form
$$
\begin{bmatrix}
    \mL_1 &  & 0 &*\\
    & \ddots & & *\\
    0 & & \mL_k & * \\
    * & * & * & (k+1)\mI
\end{bmatrix}
$$
where $\mL_1,...,\mL_k$ are Laplacians with one row and column deleted.
By Kirchoff's Matrix-Tree Theorem \cite{MooreM11} their determinant is the number of spanning forests. So they are full rank (non-zero determinant) as we assumed the underlying graph is connected.
\end{proof}

Note that the interior point method from \Cref{sec:IPM} (\Cref{thm:commodity_ipm}) assumes that vectors $x,s$ are initially given and satisfy $xs\approx t$ for some $t\in\R_{>0}$. %
Thus to solve $k$-commodity flow, we must first construct these initial vectors. 
This is done by slightly modifying the graph of the $k$-commodity flow instance without substantially changing the optimal solution. 
For this modified graph, it is easy to construct the initial vectors $x,s$.
This modified graph will have different edge costs, so solving that instance would not solve the original problem. However, we can show that by running the interior point method in ``reverse'' (i.e.~increasing $t$ in each iteration) we can swap the cost vectors back to the original cost after reaching large enough $t$.
This way we obtain a centered initial point for the original cost vector.

\begin{lemma}[Initial point lemma]\label{lem:initial_point}
Suppose we are given a feasible $k$-commodity flow instance on a graph $G=(V,E)$ with demands $d_1,\ldots,d_k\in\mathbb{Z}^V$, costs $c_1,\ldots,c_k\in\mathbb{Z}^E$, capacities $u\in\mathbb{Z}_{>0}^E$, and some $\epsilon\in(0,0.1]$. Let $C=\max_{i\in[k]}\lVert c_i\rVert_\infty$ and $U=\max(\lVert u\rVert_\infty,\max_{i\in[k]}\lVert d_i\rVert_\infty)$.

We can construct a $k$-commodity flow instance on a modified graph $G'=(V',E')$ with demands $d'_1,...,d'_k\in\mathbb{Z}^{V'}$, costs $c_1',...,c_k'\in\mathbb{Z}^{E'}$ and second set of costs $c_1'',...,c_k''\in\mathbb{Z}^{E'}$, and capacities $u'\in\mathbb{Z}_{>0}^{E'}$ with the following properties:
\begin{itemize}
    \item $|V'|=|V|+1$, $|E'|=|E|+|V|$,
and $c_1',\ldots,c_k'$, $c_1'',\ldots,c_k''$, and $u'$ have entries bounded in magnitude by $\text{poly}(m,C,U,k,\epsilon^{-1})$.

    \item For cost $c'$, we can find a primal feasible $x'\in\mathbb{R}^{(k+1)E'}$, dual feasible $y'\in\mathbb{R}^{kV'+E'}$, and corresponding slack $s\in\mathbb{R}^{(k+1)E'}$, all of whose entries are bounded in magnitude by $\text{poly}(m,C,U,k)$, such that $x's' \approx_\epsilon 1$.

    \item Further, for any feasible $x'',s''\in\R^{(k+1)E'}$ with $x''s''\approx_\epsilon \mu$
for $\mu > 30mk(CU)^2/\epsilon^3$ and cost $c'$,
we have $x''s''' \approx_{2\epsilon} \mu$ where $s'''$ is the slack when replacing costs $c'$ by $c''$. In particular, the solutions stay feasible and centered when replacing the cost vector.

    \item At last, any primal feasible $x' \in \R^{(k+1)E'}$ with $c''^\top x' \le OPT(G',c'')+\epsilon$
can be truncated to form $x\in\R^{(k+1)E}$ with
\[c^\top x\le OPT(G,c)+\epsilon\]
and for incidence matrix $\mB$ of $G$
\[\sum_{i\in[k]}\lVert \mB x_i-d_i\rVert_1\le\epsilon.\]
\end{itemize}
\end{lemma}

\begin{proof}
    Given graph $G=(V,E)$, we construct $G'=(V',E')$ by adding a new vertex with demand zero and connecting it with every other vertex with directed edges in both directions. Let $d' \in \R^{V+1}$ be the corresponding demand vector, i.e.~$d'_v = d_v$ for all $v\in V$ and $d'_u = 0$ for the newly added vertex.

    \paragraph{Initial points $x'$, $s'$:}
    A flow $x'$ for $G'$ can be constructed by routing flow $u_e/(k+1)$ on each edge $e \in E$ for each commodity.
    The flow $x'$ can be made feasible (i.e.~satisfy the demands) by routing the missing flow through the newly added vertex.
    We can assume that for each commodity and each newly added edge, it routes at least $1$ unit of flow (i.e.~we simply route $1$ unit back and forth).
    The capacity on the newly added edges $e\in E'\setminus E$ is set to $1$ plus the amount of flow on $e$ in $x'$. (So the slack of the capacity constraints is $1$ on the new edges and $u_e/(k+1)$ on the original edges.) 

    Let $x'_1,...,x'_k \in \R^{E'}$ be the flows of the individual commodities. Let $x'_{k+1}$ be the slack of the capacity constraints.
    Then $x' \in \R^{(k+1)E'}_{>0}$ is a feasible solution for the LP representing the multi-commodity flow on $G'$.
    The cost $c' \in \R^{(k+1)E'}$ is set to $1/x'$, so dual solution $y'=0$ has slack $s'=1/x'$ and thus $x's'=1$ is a centered solution.

    \paragraph{Switching the cost vector $c'$ to $c''$:}
    Now consider any feasible solution $(x'',s'')$ for the multi-commodity flow problem on $(G,c',d')$ with $x''s''\approx_\epsilon \mu$ for $\mu > 10ZU/\epsilon$ where $Z =3mkCU/\epsilon^2$. 
    Assume we replace cost $c'$ by cost $c''$ which we defined as $c''_e = c_e$ for $e\in E$ and $c''_e = Z$ for $e\in E' \setminus E$.
    Then the slack $s''$ becomes some slack $\hat{s}''_e=s''_e + (c_e - 1/x'_e)$ for $e\in E$  and $\hat{s}''_e = s''_e + Z - c'_e = s'_e + Z - 1/x'_e$ for $e \in E' \setminus E$.
    This implies $\|\hat{s}''-s''\|_\infty \le \max\{ Z, \|c\|_\infty+k+1\} \le Z$ as we can assume $x'_e \ge 1/(k+1)$ by construction for integral capacities $u$. 
    By $x'' \le U$ we have that 
    $s'' \ge 0.5\mu/x'' \ge 0.5\mu/U \ge Z/\epsilon$ in particular, $\hat{s}'' > 0$ is a feasible slack, so the dual solution is still feasible. 
    Further we have
    $x''~\hat{s}'' = x''s'' + x''(\hat{s}''-s'')$
    where by $x'' \le U$ and $\|\hat{s}''-s''\|_\infty \le Z$
    and $\mu > 10UZ/\epsilon$ we have $x''\hat{s}'' \approx_{2\epsilon} \mu$. 

    \paragraph{Transforming solution of $G',c''$ to $G,c$:}
    At last, let $\ox\in\R^{(k+1)E'}$ be any primal feasible solution with $c''^\top \ox \le OPT(G',c'') + \epsilon$.
    Let $x\in\R^{(k+1)E}$ be the restriction of $\ox$ onto the original edges $E$.
    Since $c''$ is the same as $c$ except for the newly added edges, which have positive cost, we have
    $$c^\top x \le c''^\top \ox \le OPT(G', c'') + \epsilon \le OPT(G, c) + \epsilon$$
    where the last inequality comes from the fact that any feasible solution on $G$ is also feasible on $G'$ and both graphs share the same costs except for the newly added edges.
    The solution $x$ is almost feasible because
    \begin{align*}
        c''^\top \ox 
        = 
        c^\top x + Z \sum_{i\in[k]} \sum_{e\in E'\setminus E} (\ox_i)_e
        \ge
        -mCU + Z \sum_{i\in[k]} \sum_{e\in E'\setminus E} (\ox_i)_e
    \end{align*}
    which together with 
    $$c''^\top \ox \le OPT(G',c'') + \epsilon \le OPT(G, c) + \epsilon \le mCU + \epsilon$$
    implies
    $$
    -mCU + Z \sum_{i\in[k]} \sum_{e\in E'\setminus E} (\ox_i)_e
    \le
    mCU + \epsilon
    $$
    and thus $\sum_{i\in[k]} \sum_{e\in E'\setminus E} (\ox_i)_e \le (2mCU+\epsilon)/Z$.
    For $Z = 3mkCU/\epsilon^2$ this implies for the incidence matrix $\mB$ of graph $G$,
    $$
    \sum_{i\in[k]} \| \mB x_i - d_i\|_1
    \le
    \sum_{i\in[k]} \sum_{e\in E'\setminus E} (\ox_i)_e
    \le
    \epsilon.
    $$

\end{proof}

\subsection{Implementing the IPM}
\label{sec:algorithm:implement}

We now have all tools available to us to prove our main results \Cref{thm:intro:mincost,thm:intro:throughput}.
We start by proving that the interior point method from \Cref{sec:IPM} (\Cref{alg:commodity_ipm}) can be implemented efficiently. \Cref{thm:ipm_implement} states the complexity of \Cref{alg:commodity_ipm} when implemented via the data structure from \Cref{sec:vector}.

\begin{theorem}\label{thm:ipm_implement}
    For any $0\le\mu\le1$, we can implement \Cref{alg:commodity_ipm} such that if the input is an incidence matrix, the total time is bounded by
    $$
    \tO(\sqrt{m}(n^{\omega-1/2}+n^{\omega(1,1,\mu)-\mu/2}+n^{1+\mu}+n\log \frac{Ut^\init}{S t^\target})\log|t^\init/t^\target|)
    $$
    where $U$ is the largest capacity and $S$ is the largest entry of $s^\init$.
    For current bounds on $\omega\approx2.373$ and $\alpha\le0.319$ this is
    $$
    \tO(\sqrt{m}(n^{\omega-1/2}+n\log \frac{Ut^\init}{S t^\target})\log|t^\init/t^\target|)
    $$
\end{theorem}

To prove \Cref{thm:ipm_implement}, we need one more data structure for maintaining the solution of a sparsely changing linear system.

\begin{lemma}\label{lem:inverse_maintenance}
    There exists a deterministic data structure with the following operations:
    \begin{itemize}
        \item \textsc{Initialize} %
        Initializes on given $\mM\in\R^{d\times d}, v\in\R^d$ and returns $\mM^{-1}v$ in $O(d^\omega)$ time.
        \item \textsc{Update}($\mU,\mV,v$) %
        For any $0\le \mu$, perform a rank $d^\mu$ update $\mM \leftarrow \mM + \mU\mV^\top$ and replace $v$ by the given new vector.
        Then return $\mM^{-1}v$ in $O(d^{\omega(1,1,\mu)})$ time.
        \item \textsc{TempUpdate} %
        For any $0\le\mu\le 1$, \emph{temporarily} perform a rank $d^\mu$ update $\mM \leftarrow \mM + \mU\mV^\top$ where $\mU,\mV$ have at most $d^\nu$ non-zero entries, and change up to $d^\nu$ entries of $v$.
        Return $\mM^{-1} v$ in $O(d^{\omega\cdot\mu}+d^{1+\nu})$ time. 
        Then revert these changes to $\mM$ and $v$ again.
    \end{itemize}
\end{lemma}

\begin{proof}
    First, observe that
    $$
    \mN = \begin{bmatrix}
        \mM & v \\
        0 & -1
    \end{bmatrix}
    \text{ has inverse }
    \mN^{-1} = \begin{bmatrix}
        \mM^{-1} & \mM^{-1}v \\
        0 & -1
    \end{bmatrix}.
    $$
    So to maintain $\mM^{-1}v$ we just need to maintain the inverse of $\mN$ and return the last columns of its inverse.
    
    To maintain the inverse, we use the Woodbury-identity \cite{Woodbury50,ShermanM50}:
    $$
    (\mN^{-1} + \mU \mV^\top)^{-1} = \mN^{-1} - \mN^{-1}\mU (\mI + \mV^\top \mN^{-1} \mU)^{-1} \mV^\top \mN^{-1}
    $$
    For \textsc{Update} of rank $d^\mu$, the matrices $\mU,\mV$ have $d^\mu$ columns.
    If $\mu \le 1$, the Woodbury identity allows us to obtain the new inverse in $O(n^{\omega(1,1,\mu)})$ time.
    For $\mu \ge 1$, we first compute $\mU\mV^\top$ in $O(n^{\omega(1,1,\mu)})$ time, and then compute the inverse without Woodbury identity in $O(d^\omega)$ time.

    For \textsc{TempUpdate} of rank $d^\mu$ where $\mU,\mV$ have at most $d^\nu$ non-zero entries, we can compute $\mN^{-1}\mU$ and $\mV^\top \mN^{-1}$ in $O(d^{1+\nu})$ time, and then compute $(\mI + \mV^\top \mN^{-1} \mU)^{-1}$ in $O(d^{\mu\omega})$ time.
    To return the last columns of the inverse, we must compute $$
    (\mN^{-1} - \mN^{-1}\mU (\mI + \mV^\top \mN^{-1} \mU)^{-1} \mV^\top \mN^{-1})e_{d+1}
    $$
    which takes $O(d^{1+\tau})$ time since $\mN^{-1}\mU$ and $\mV^\top \mN^{-1}$ are of size $d\times d^\mu$ and $d^\mu\times d$ respectively.
\end{proof}

\begin{proof}[Proof of \Cref{thm:ipm_implement}]
    \Cref{alg:commodity_ipm} takes $\tilde{O}(\sqrt{km}\log |t^\init/t^\target|)$ iterations (see \Cref{thm:commodity_ipm}). 
    We split this into batches of $O(\sqrt{m})$ iterations. 
    In the following, we analyze the complexity of one such batch, 
    the overall complexity then increases by a $\tilde{O}(\sqrt{k}\log |t^\init/t^\target|)$ factor.

    \paragraph{Maintaining $\ox\approx x,\os\approx s$}
    The main task is to maintain the approximate vectors $\ox,\os$.
    We maintain these via the data structure \Cref{thm:primaldual:maintenance}
    where we use $\beta = \frac{t'}{32\lambda\|g\|_2}$, $z_j = \omS^{-1}_j g_j$, and accuracy $\epsilon = \lambda/500$.
    Thus the vectors maintained by \Cref{thm:primaldual:maintenance} are exactly as given in \Cref{line:commodityipm:deltaxdeltas} of \Cref{alg:commodity_ipm}.
    Note that whenever some entry of any output $\ox_j,\os_j$ changes, we must update an entry of $z_j$ for the same $j$, and an entry of $w$. This takes $\hO(k)$ time per update by \Cref{thm:primaldual:maintenance}.
    By \Cref{lem:vector:amortizednumchanges} there are only $\tilde{O}(m)$ changes over $\sqrt{m}$ iterations, so the total time of all calls to \textsc{Update} of \Cref{thm:primaldual:maintenance} is bounded by $\hO(k m)$.
    The total time of the $\sqrt{m}$ calls to \textsc{Add} are bounded by $\tilde{O}(k^2m^{1+o(1)} + Tkn\log W)$.
    Here $W$ is a bound on the ratio of largest to smallest entry in any $d_i/(d_\Sigma\os_j)$ for $i\neq j$. 
    We will bound $W$ at the end of this proof.
    
    \paragraph{Maintaining $\ov$ and $g$}
    We use $\ov = \ox\os/t^\init$ where $t^\init$ is the value of $t$ at the start of an $O(\sqrt{m})$ iteration batch.
    By the small number of iterations, we know $t$ can change by at most some $1+O(1/\lambda)$ factor.
    Further, $\ox\approx_{\lambda/500}x$ and $\os\approx_{\lambda/500} s$ so we have $\|\ov - xs/t\|_\infty \le 1/(48\lambda)$.
    It is easy to maintain $\ov$: whenever an entry of $\ox$ or $\os$ changes, we change the respective entry in $\ov$.
    The vector $g = \nabla \Phi(\ov)$ can be just as easily maintained since the $i$th entry of $g$ depends only on the $i$th entry of $\ov$.
    This also allows us to maintain $\|g\|_2$.

    In summary, the time complexity of maintaining $\ov,g,\|g\|_2$ is $O(1)$ per changed entry of $\ox$ or $\os$.
    By \Cref{lem:vector:amortizednumchanges} there are at most $\tilde{O}(m)$ entry changes to $\ox,\os$ in total, as we consider only a sequence of $O(\sqrt{m})$ iterations.
    So the total time is bounded by $\tilde{O}(m)$.

    \paragraph{Computing $v_1,v_2,...,v_k$}
    The only value that is left to implement the IPM are the vectors $v_1,...,v_k$ that must be given as input to \textsc{Add} of \Cref{thm:primaldual:maintenance}.
    By \Cref{alg:commodity_ipm}, these values are given by
    $$
    \begin{bmatrix}
        v_1 \\ \vdots \\ v_k
    \end{bmatrix}
    =
    \mE^{-1}
    \underbrace{\begin{bmatrix}
        \mA^\top (\omS_1^{-1} g_1 - \mD_1 w) \\
        \vdots \\
        \mA^\top (\omS_k^{-1} g_k - \mD_k w) 
    \end{bmatrix}}_{=:b}
    $$
    where
    \begin{align}
    \mE := 
    \left[\begin{array}{ccc}
    \mA^\top \mD_1 \mA &  &   0  \\
     & \ddots & \\
    0   &  & \mA^\top\mD_k\mA 
    \end{array}\right]-
    \left[\begin{array}{c}
    \mA^\top \mD_1\\
    \vdots \\
    \mA^\top\mD_k\\
    \end{array}\right]
    \mD_\Sigma^{-1}
    \left[\begin{array}{ccc}
    \mD_1\mA  &  \cdots & \mD_k\mA
    \end{array}\right]
    \label{eq:alg:defineE}
    \end{align}
    Note that changing one entry change to $\ox,\os$ (and thus one change to $\mD_\Sigma$ and one change some $\mD_j$) changes the matrix $\mE$ as follows:
    The left block-diagonal matrix changes in only $O(1)$ entries, while the matrix on the right of \eqref{eq:alg:defineE} results in a rank $O(1)$ update that changes $O(k^2)$ entries.
    In particular, we can phrase the update are some $\mE\leftarrow\mE+\mU\mV^\top$ where $\mU,\mV$ have $O(k)$ columns and $O(k)$ non-zeros.
    Thus we can maintain the vectors $v_1,...,v_k$ via the data structure of \Cref{lem:inverse_maintenance}.
    For that, we call \textsc{TempUpdate} of \Cref{lem:inverse_maintenance} in each iteration to solve the current linear system. Once the total rank of all past updates exceeds $n^\mu$, we call \textsc{Update}.
    If the total change is of rank more than $nk$, we reinitialize the data structure from scratch at cost $O((nk)^\omega)$.
    Since at most $\tilde{O}(2^{2i})$ entries of $\ox,\os$ change every $2^i$ iterations (for $i=0,...,\log\sqrt{m}$) by \Cref{lem:vector:amortizednumchanges} we can bound the amortized complexity by
    \begin{align*}
    &~
    \tilde{O}(
        (kn)\cdot n^\mu \cdot k
        +\sum_{i=\mu\log \sqrt{n}}^{\log\sqrt{nk}} (kn)^{\omega(1,1,2\log_{kn}(2^i))}/2^i
        +\sum_{i=\log\sqrt{nk}}^{\log m} (nk)^\omega/2^i
    ) \\
    =&~
    \tilde{O}(
        k^2 n^{1+\mu}
        +\sum_{i=\mu\log \sqrt{n}}^{\log\sqrt{nk}} k^2n^{\omega(1,1,2\log_{n}(2^i))}/2^i
        +\sum_{i=\log\sqrt{nk}}^{\log m} (nk)^\omega/2^i
    ) \\
    =&~
    \tilde{O}(
        k^2 \left(
        n^{1+\mu}
        +n^{\omega(1,1,\mu)-\mu/2}
        \right)
        +(nk)^{\omega-1/2}
    ) 
    \end{align*}
    where we use the convexity of the matrix exponent $\omega(1,1,\cdot)$, so the maximum is bounded by the two end-points $n^{\mu/2} \le 2^i \le \sqrt{nk}$.
    \paragraph{Overall complexity}
    At the end of one $O(\sqrt{m})$ iteration batch, we compute the current value of $x,s$ explicitly, which takes $\hO(k^2m)$ time by \Cref{thm:primaldual:maintenance}. 
    At the start of the next batch, we reinitialize the data structures of \Cref{thm:primaldual:maintenance} again, which also takes $\hO(k^2m)$ time.
    Overall, such a batch of $\tilde{O}(\sqrt{m})$ iterations takes 
    $\tilde{O}(k^2m^{1+o(1)}+\sqrt{m}k^2(n^{\omega-1/2}+n^{1+\mu}+n^{\omega(1,1,\mu)-\mu/2}+n\log W))$ time,
    where $W$ is the ratio of largest to smallest entry in $d_i/(d_\Sigma\os_j)$ for $i\neq j$.
    We will bound this ratio via \Cref{cor:bitlengthbound}.
    As the maximum number of iterations is bounded by $\tilde{O}(\sqrt{km}|t^\init/t^\target|)$, we can bound $\log W$ by $\tilde{O}(\log \frac{Ut^\init}{S t^\target})$ via \Cref{cor:bitlengthbound} where $U = \|U\|_\infty$ is a bound on the largest capacity, and $S = \|s^\init\|_\infty$.

    In summary, the complexity of the algorithm is bounded by 
    $$\tilde{O}(k^{2.5}\sqrt{m}(n^{\omega-1/2}+n^{1+\mu}+n^{\omega(1,1,\mu)-\mu/2}+n\log\frac{Ut^\init}{S t^\target})\log|t^\init/t^\target|)$$
    which for current bounds on $\omega$ is 
    $$\tilde{O}(k^{2.5}\sqrt{m}(n^{\omega-1/2}+n\log \frac{Ut^\init}{S t^\target})\log|t^\init/t^\target|).$$
\end{proof}

By combining \Cref{thm:ipm_implement} with the construction of an initial point (\Cref{lem:initial_point}) we prove our two main results.

\throughput*
\mincost*
\begin{proof}[Proof of \Cref{thm:intro:mincost,thm:intro:throughput}]
We start with the min-cost version \Cref{thm:intro:mincost}.
We transform the given $k$-commodity flow instance $(G,c,d)$ via \Cref{lem:initial_point} to a $k$-commodity flow instance $(G',c',d')$ to obtain an initial point $x',s',t$ with $\Phi(x's'/t) \le 16n$. 
This requires $\epsilon$ in \Cref{lem:initial_point} to be at most $1/16\lambda = O(1/\log m)$ for $\lambda$ as defined in \Cref{alg:commodity_ipm}. 
At the end, we want to reconstruct a solution for the original instance $(G,c,d)$ that is at most some additive $\delta>0$ away from the optimal solution, so by \Cref{lem:initial_point} it suffices to set $\epsilon = \min\{\delta,1/(16\lambda)\}$.

We first run the algorithm \Cref{alg:commodity_ipm} in reverse (i.e.~increase $t$ in every iteration) $t \ge 3m(CU)^2/\epsilon^2$. Let $x'', s''$ be the vectors obtained at the end.
By \Cref{lem:initial_point} we can now replace the cost vector $c'$ by $c''$ and the solutions stay feasible and centered.
So we now run the algorithm \Cref{alg:commodity_ipm} again until $t = O(\epsilon/m)$ so that by $xs\approx_{1/16}t$ we know the solution is at most some $\epsilon$ away from the optimal cost of the modified $k$-commodity flow instance $(G',c'',d')$.
By \Cref{lem:initial_point}, this is good enough to obtain an approximate solution of the original $k$-commodity flow instance $(G,c,d)$.

The time complexity is $\tilde{O}(\sqrt{m}(n^{\omega-1/2}+n\log(UC/\epsilon))\log(UC/\epsilon))$ by \Cref{thm:ipm_implement} because $t,x,s$ are all bounded by $\poly(mCU/\epsilon)$ by \Cref{lem:initial_point} and \Cref{lem:ratiobound}.

\paragraph{Throughput version}

To solve the throughput version, we simply add an edge for each target-sink pair of negative cost $-1$ and set the cost on all other edges to $0$. 
The demand on each vertex is also set to $0$. 
The maximum throughput version is now a minimum cost version and can be solved in $\tilde{O}(\sqrt{m}(n^{\omega-1/2}+n\log(UC/\epsilon))\log(UC/\epsilon))$ time.
Note that the demands are not be perfectly satisfied by the computed solution. 
To fix this, we can route any superfluous flow back to its origin, e.g.~by solving $k$ single-commodity flow instances. This can reduce the maximum throughput by at most an extra $\epsilon$.
\end{proof}

\section*{Acknowledgement}

We would like to thank Yin Tat Lee, Aaron Sidford, Rasmus Kyng, and Richard Peng for helpful discussions.
This work was partially done while %
Jan van den Brand was at UC Berkeley and the Max Planck Institute for Informatics.
Part of this research was funded by ONR BRC grant N00014-18-1-2562, and by the Simons Institute for the Theory of Computing through a Simons-Berkeley Postdoctoral Fellowship, and by the Max Planck Institute.
\bibliographystyle{alpha}
\bibliography{ref}
\newpage

\appendix

\section{Appendix}
\label{sec:appendix:schur}

\schurcomplement*

\begin{proof}%
To show that $D$ is positive definite, note that for all $x\in\mathbb{R}^n\setminus\{0\}$,
\[x^TDx=\begin{pmatrix}0\\x\end{pmatrix}^T\begin{pmatrix}A&B\\C&D\end{pmatrix}\begin{pmatrix}0\\x\end{pmatrix}>0\]

To show that $E=A-BD^{-1}C$ is positive definite, note that
\[\begin{pmatrix}I&-X\\0&I\end{pmatrix}\]
is full rank, so
\[\begin{pmatrix}I&-X\\0&I\end{pmatrix}\begin{pmatrix}A&B\\C&D\end{pmatrix}\begin{pmatrix}I&0\\-X^T&I\end{pmatrix}=\begin{pmatrix}A-XC-BX^T+XDX^T&B-XD\\C-DX^T&D\end{pmatrix}\]
is positive definite.

Setting $X=BD^{-1}$, we get that the submatrix $A-XC-BX^T+XDX^T=A-BD^{-1}C$ (recall that $B=C^T,D=D^T$) is positive definite as well.

Finally, to show \Cref{eq:structure:psd_inverse}, we have
\begin{align*}
    &=\begin{pmatrix}I&0\\-D^{-1}C&I\end{pmatrix}\begin{pmatrix}E^{-1}&0\\0&I\end{pmatrix}\begin{pmatrix}I&-B\\0&I\end{pmatrix}\begin{pmatrix}A&B\\D^{-1}C&I\end{pmatrix}\\
    &=\begin{pmatrix}I&0\\-D^{-1}C&I\end{pmatrix}\begin{pmatrix}E^{-1}&0\\0&I\end{pmatrix}\begin{pmatrix}A-BD^{-1}C&0\\D^{-1}C&I\end{pmatrix}\\
    &=\begin{pmatrix}I&0\\-D^{-1}C&I\end{pmatrix}\begin{pmatrix}E^{-1}&0\\0&I\end{pmatrix}\begin{pmatrix}E&0\\D^{-1}C&I\end{pmatrix}\\
    &=\begin{pmatrix}I&0\\-D^{-1}C&I\end{pmatrix}\begin{pmatrix}I&0\\D^{-1}C&I\end{pmatrix}\\
    &=\begin{pmatrix}I&0\\0&I\end{pmatrix}\\
\end{align*}
\end{proof}

\section{Extension of Vector-Maintenance}

The proof of \Cref{lem:vector:maintenance} follows the approach of \cite{BrandLSS20,BrandLN+20,BrandLL+21}. 
We can not directly use their data structures, because for them the vector $\delta_s$ (which is added to $s$ in each iteration) was of the form $\mA v$, whereas our $\delta_s$ in \Cref{lem:vector:maintenance} is of the form $\mG \mA v + w\beta$ for some diagonal matrix $\mG$, vector $w$ and scalar $\beta$. 
So we must prove that their approach can be extended to this more general shape of $\delta_s$.

\vectormaintenance*

We prove \Cref{lem:vector:maintenance} in two parts:
(i) We describe data structures that maintain the sums in some implicit form, so that we can query any entry of $s$ efficiently.
(ii) We detect which entries of $s$ might have changed a lot since the last time we set $\os_i \leftarrow s_i$, i.e.~we detect indices $i$ where $|\os_i - s_i| \le \epsilon_i$ might not hold anymore.
For these indices we then set $\os_i \leftarrow s_i$ where $s_i$ can be computed via the implicit representation.

Part (i) is proven in \Cref{sec:vector:subroutines} and part (ii) (which concludes the proof of \Cref{lem:vector:maintenance}) is proven in \Cref{sec:vector:maintain}

\subsection{Subroutines}
\label{sec:vector:subroutines}

To maintain the vector $s$ in implicit form, we split the sum
$$
    s^{(t)} := s^{(0)} + \sum_{\ell=1}^t \left(\mD^{(\ell)} \mA h^{(\ell)} + \beta^{(\ell)}w^{(\ell)}\right)
    $$
into
$$
    \sum_{\ell=1}^t \mD^{(\ell)} \mA h^{(\ell)} \text{  and   }
    \sum_{\ell=1}^t\beta^{(\ell)}w^{(\ell)}
$$
The first sum can be maintained implicitly via the following \Cref{lem:vector:implicit:sumofproduct}.
The second sum is maintained via \Cref{lem:vector:implicit:sumofvector}.

\begin{lemma}\label{lem:vector:implicit:sumofproduct}
    There exists a deterministic data structure with the following operations
    \begin{itemize}
        \item \textsc{Initialize}$(\mA\in\R^{m\times n},d\in\R^m)$
        Initialize on matrix $\mA$ and vector $d$ in $O(\nnz(\mA))$ time.%
        For all $i\in[m]$, let $\nnz(a_i)$ be the number of nonzero entries in the $i\th$ row of $\mA$.
        \item \textsc{Update}$(i\in[m], c\in\R)$
        Set $d_i \leftarrow c$ in $O(\nnz(a_i))$ time.
        \item \textsc{Add}$(h\in\R^n)$
        Store vector $h$ in $O(n)$ time.
        \item \textsc{Query}$(i\in[m])$
        Let $h^{(\ell)}$ be the vector $h$ given during the $\ell\th$ call to \textsc{Add}.
        Let $d^{(\ell)}$ be the state of $d$ during the $\ell\th$ call to \textsc{Add}.
        Let 
        $$
        s^{(t)} := \sum_{\ell=1}^t \mD^{(\ell)} \mA h^{(\ell)}$$
        Return $s^{(t)}_i$ in $O(\nnz(a_i))$ time, where $t$ is the number of calls to \textsc{Add} so far.
    \end{itemize}
\end{lemma}

\begin{proof}
The data structure stores explicitly:
\begin{itemize}
    \item $\mA\in\mathbb{R}^{m\times n}$ as a sparse matrix
    \item $d\in\mathbb{R}^m$, the diagonal of $\mD$
    \item $T\in\mathbb{Z}_{\ge 0}$, the current iteration (number of $h$s)
    \item For each $t\in\{0,1,\ldots,T\}$, $\tilde{h}^{(t)}=\sum_{l=1}^t h^{(l)}\in\mathbb{R}^n$, the prefix sums of $h^{(t)}$
    \item $t\in\mathbb{R}^n$ where $t_i$ is the most recent iteration where $d_i$ was changed %
    \item $\hat{s}\in\mathbb{R}^m$ where $\hat{s}_i=\sum_{l=1}^{t_i-1}(\mD^{(l)}\mA h^{(l)})_i$
\end{itemize}

\begin{algorithm2e}[t]
\caption{Sum Of Product. \label{alg:vector:sumofproduct}}
\SetKwProg{Proc}{procedure}{}{}
\Proc{\textsc{Init}$(\mA\in\mathbb{R}^{m\times n},d\in\mathbb{R}^n)$}{
Initialize $\mA$ and $d$\\
$T\gets 0$\\
$t\gets 1$ (entrywise)\\
$\tilde{h}^{(0)}\gets 0$\\
$\hat{s}\gets 0$
}
\Proc{\textsc{Update}$(i\in[m],c\in\mathbb{R})$}{
$\hat{s}_i\gets \text{Query}(i)$\\
$t_i\gets T+1$\\
$d_i\gets c$
}
\Proc{\textsc{Add}$(h\in\mathbb{R}^n)$}{
$\tilde{h}^{(T+1)}\gets \tilde{h}^{(T)}+h$\\
$T\gets T+1$
}
\Proc{\textsc{Query}$(i\in[m])$}{
Return $\hat{s}_i+d_i (e_i^\top \mA)(\tilde{h}^{(T)}-\tilde{h}^{(t_i-1)})$
}
\end{algorithm2e}

The invariants hold when the data structure is initialized, and continue to hold after each operation.

$\textsc{Query}$ returns the right result because
\begin{align*}
    (s^{(T)})_i&=(\sum_{l=1}^{T}\mD^{(l)}\mA h^{(l)})_i\\
    &=(\sum_{t=1}^{t_i-1}\mD^{(l)}\mA h^{(l)})_i+(\sum_{t_i}^{T}\mD^{(l)}\mA h^{(l)})_i\\
    &=\hat{s}_i^{(T)}+(\sum_{t_i}^{T}\mD^{(T)}\mA h^{(l)})_i\\
    &=\hat{s}_i^{(T)}+\mD^{(T)}_{ii}(\mA \sum_{t_i}^{T}h^{(l)})_i\\
    &=\hat{s}_i^{(T)}+d_i^{(T)}(\mA(\tilde{h}^{(T)}-\tilde{h}^{(t_i-1)}))_i\\
    &=\hat{s}_i^{(T)}+d_i^{(T)}(e_i^\top \mA)(\tilde{h}^{(T)}-\tilde{h}^{(t_i-1)})
\end{align*}
The time complexity of each operation is evident from the pseudocode.
\end{proof}

\begin{lemma}\label{lem:vector:implicit:sumofvector}
    There exists a deterministic data structure with the following operations
    \begin{itemize}
        \item \textsc{Initialize}$(w \in \R^m)$
        Initialize on the vector $w$ in $O(m)$ time.
        \item \textsc{Update}$(i\in[m], c\in\R)$
        Set $w_i \leftarrow c$ in $O(1)$ time.
        \item \textsc{Add}$(\beta\in\R)$
        Store scalar $\beta$ in $O(1)$ time.
        \item \textsc{Query}$(i\in[m])$
        Let $\beta^{(\ell)}$ be the scalar $\beta$ given during the $\ell\th$ call to \textsc{Add}.
        Let $w^{(\ell)}$ be the state of $w$ during the $\ell\th$ call to \textsc{Add}.
        Let 
        $$
        s^{(t)} := \sum_{\ell=1}^t \beta^{(\ell)} w^{(\ell)}$$
        Return $s^{(t)}_i$ in $O(1)$ time, where $t$ is the number of calls to \textsc{Add} so far.
    \end{itemize}
\end{lemma}

\begin{proof}
The data structure stores explicitly:
\begin{itemize}
    \item $T\in\mathbb{Z}_{\ge 0}$, the current iteration (number of $w$s)
    \item For each $t\in\{0,1,\ldots,T\}$, $\tilde{w}^{(t)}=\sum_{l=1}^t w^{(l)}\in\mathbb{R}^m$, the prefix sums of $w^{(t)}$, 
    \item $t\in\mathbb{R}^m$ where $t_i$ is the most recent iteration where $w_i$ was changed. %
    \item $\hat{s}\in\mathbb{R}^m$ where $\hat{s}_i=\sum_{l=1}^{t_i-1}(\beta^{(l)}w^{(l)})_i$
\end{itemize}

\begin{algorithm2e}[t]
\caption{Sum Of Vector. \label{alg:vector:sumofvector}}
\SetKwProg{Proc}{procedure}{}{}
\Proc{\textsc{Init}$(w\in\mathbb{R}^{m})$}{
$T\gets 0$\\
$t\gets 1$ (entrywise)\\
$\tilde{w}^{(0)}\gets 0$\\
$\hat{s}\gets 0$
}
\Proc{\textsc{Update}$(i\in[m],c\in\mathbb{R})$}{
$\hat{s}_i\gets \text{Query}(i)$\\
$t_i\gets T+1$\\
$w_i\gets c$
}
\Proc{\textsc{Add}$(\beta\in\mathbb{R})$}{
$\tilde{w}^{(T+1)}\gets \tilde{w}^{(T)}+w$\\
$T\gets T+1$
}
\Proc{\textsc{Query}$(i\in[m])$}{
Return $\hat{s}_i+\beta (\tilde{w}^{(T)}_i-\tilde{w}^{(t_i-1)}_i)$
}
\end{algorithm2e}

The invariants hold when the data structure is initialized, and continue to hold after each operation.

$\textsc{Query}$ returns the right result because
\begin{align*}
    (s^{(T)})_i&=(\sum_{l=1}^{T}\beta^{(l)} w^{(l)})_i\\
    &=(\sum_{l=1}^{t_i-1}\beta^{(l)} w^{(l)})_i+(\sum_{l=t_i}^{T}\beta^{(l)} w^{(l)})_i\\
    &=\hat{s}^{(T)}_i+(\sum_{l=t_i}^{T}\beta^{(T)} w^{(l)})_i\\
    &=\hat{s}^{(T)}_i+\beta^{(T)} (\sum_{l=t_i}^{T}w^{(l)})_i\\
    &=\hat{s}^{(T)}_i+\beta^{(T)} (\tilde{w}^{(T)}-\tilde{w}^{(t_i-1)})_i
\end{align*}

The time complexity of each operation is evident from the pseudocode. %

\end{proof}

To detect which entries of $s$ change a lot from one iteration to the next, we use the following data structure by \cite{BrandLN+20}. This data structure is the only graph theoretic tool used in our algorithm and internally relies on the expander decomposition technique \cite{NanongkaiS17,NanongkaiSW17,WulffNilsen17,SaranurakW19,HuaKGW22,BrandGJLLPS22}.
In \cite{BrandLN+20} a randomized dynamic expander decomposition was used, but deterministic variants exist as well \cite{ChuzhoyGLNPS20}. The randomized variant would need $\tilde{O}(1)$ time per \textsc{Scale}, whereas the deterministic variant is only subpolynomial.
\begin{lemma}[\cite{BrandLN+20}]\label{lem:vector:heavyhitter}
    There exists a deterministic data structure with the following operations
    \begin{itemize}
        \item \textsc{Initialize}$(\mA\in\R^{m\times n},g\in\R^m_{\ge0})$
        Initialize on the given incidence matrix $\mA$, edge weights $g$ in $\hO(m)$ time.
        \item \textsc{Scale}$(e\in [m], \delta\in\R_{\ge0})$ Set $g_e \leftarrow \delta$ in $\hO(1)$ amortized time.
        \item \textsc{Query}$(h\in\R^n,\epsilon>0)$
        Return all indices $I\subset[m]$ with $|g_i (\mA h)_i| > \epsilon$ in $\tO(\|\mG \mA h\|_2^2/\epsilon^2 + n \log W)$ time, where $W$ is the ratio of largest to smallest non-zero entry in $g$.
        At most $\tO(\|\mG \mA h\|_2^2/\epsilon^2)$ indices are returned.
    \end{itemize}
\end{lemma}

We remark that when we use these data structures, matrix $\mA$ is actually not an incidence matrix. 
Instead, $\mA$ is an incidence matrix \emph{of which one columns was deleted} (\Cref{lem:invertible}). 
The above data structure still works in that case, because after deleting a column of an incidence matrix, the remaining matrix is an incidence matrix where at most $O(n)$ rows contain just a single $\pm 1$ entry. 
The indices of large entries $|g_i (\mA h)_i| > \epsilon$ can be returned trivially for those rows by just checking $|g_i h_j| > \epsilon$.

\subsection{Proof of \Cref{lem:vector:maintenance}}
\label{sec:vector:maintain}

We now prove \Cref{lem:vector:maintenance} using the data structure from the previous \Cref{sec:vector:subroutines}.
The proof-idea is that if an entry $s_i$ changed by $\Omega(\epsilon_i)$ for some $i$, then we set $\os_i \leftarrow s_i$.
Since other entries $s_j$ did not change by some $\Omega(\epsilon_i)$, the old value of $\os_j$ is still a valid approximation.
This way we can maintain $\os$ in sublinear time per iteration, because we only change a few entries in each iteration.

Note that an entry $s_i$ might also change by some $\Omega(\epsilon_i)$ over a longer time interval but only a little in each iteration. 
To maintain $\os$ with $|\os_i - s_i| < \epsilon_i$ for these slowly changing entries, we do the following:
Every $2^\ell$ iterations, we detect all entries $i$ where $s_i$ changed by some $\Omega(\epsilon_i/\log m)$ over the past $2^\ell$ iterations.
This is done for all $0\le \ell \le \log \sqrt{m}$.
If an entry changed sufficiently, we set $\os_i \leftarrow s_i$.
This way we make sure that $\os_i$ is always at most $\epsilon_i$ away from $s_i$, because $s_i$ can never change too much without our data structure updating $\os_i$.

\begin{algorithm2e}[t]
\caption{Vector Maintenance. \label{alg:vector:maintenance}}
\SetKwProg{Proc}{procedure}{}{}
\SetKwProg{Params}{parameters}{}{}
\Params{and their initial values}{
	$t = 0$, $\os=s^{(0)}$, $F^k = \emptyset$ for $k=0,...,\log\sqrt{m}$.
}
\Proc{\textsc{Update}$(i\in[m],c,c',c'')$}{
    Pass the calls to the data structures \Cref{lem:vector:implicit:sumofproduct,lem:vector:implicit:sumofvector}.\\
    $F^k \leftarrow F^k \cup \{i\}$ for $k=0,...,\log m$ \\ %
    Update $d_i,d'_i,w_i\leftarrow c,c',c''$.
}
\Proc{\textsc{SetAccuracy}$(i\in[m],\delta)$}{
    $F^k \leftarrow F^k \cup \{i\}$ for $k=0,...,\log m$ \\ %
    $\epsilon_i \leftarrow \delta$
}
\Proc{$\textsc{Add}(h,h'\in\R^n, \beta>0)$}{
	$t \leftarrow t+1$\\
	Update implicit $s^{(t)}$ via \Cref{lem:vector:implicit:sumofproduct,lem:vector:implicit:sumofvector}.\\
	\For{$k \le \log \sqrt{m}$ with $2^k$ divides $t$\label{line:vector:k}}{
		Let $g_i = d_i/\epsilon_i$, $g'_i=d'_i/\epsilon_i$ for $i \notin F^k$, $g_i = g'_i = 0$ for $i \in F^k$ \label{line:Vector:update_g}\\
		\tcp{$i\in F^k$ are indices for which we called \textsc{SetAccuracy} or \textsc{Update} over the past $2^k-1$ iterations.}
		\tcp{We now search for indices $i\notin F^k$ where we must update $\os_i$.}
		$I_k \leftarrow$ set of $i$ with $|(\mG \mA (\sum_{\ell=2^k+1}^t h^{(\ell)})_i|>1/(30\log m)$ or $|(\mG' \mA (\sum_{\ell=2^k+1}^t h'^{(\ell)})_i|>1/(30\log m)$. \label{line:vector:heavyhitter}\\
		$I_k \leftarrow I_k \cup$ set of $i \subset [n]\setminus F^k$ with 
		$|w_i\sum_{\ell=t-2^k+1}^t\beta^{(\ell)}| > \epsilon/(30\log m)$. \label{line:vector:gradient}\\
        $J_k \leftarrow I_k \cup F^k$ \\
        Set $\os_i$ to $s_i^{(t)}$ for $i\in J_k$. \label{line:vector:update_os} \\
        $F^k \leftarrow \emptyset$ 
	}
	\Return $\ox$, $\os$
}
\end{algorithm2e}

The algorithm is given by \Cref{alg:vector:maintenance}.
We first prove correctness, i.e.~that $\os$ returned by the $t$th call to \textsc{Add} satisfies $|\os_i - s_i| < \epsilon_i$.

\begin{lemma}
The output $\os \in \R^n$ 
returned by the $t$th call to \textsc{Add}
satisfies $\os \approx s^{(t)}$.
\end{lemma}

\begin{proof}
We maintain $s^{(t)}$ implicitly via \Cref{lem:vector:implicit:sumofproduct,lem:vector:implicit:sumofvector}.
We now argue that we set $\os_i \leftarrow s^{(t)}_i$ 
before $|\os_i - s^t_i| > \epsilon$ 
can occur.

Consider a loop of \Cref{line:vector:k} for some $k$.
Note that set $I_k$ contains all $i\notin F^k$ where
\begin{align}
|s^{(t)}_i - s^{(t-2^k)}_i|
=
|\mG(\mA \sum_{\ell=2^k+1}^t h^{(\ell)})_i + \mG'(\mA \sum_{\ell=2^k+1}^t h'^{(\ell)})_i + \sum_{\ell=2^k+1}^t (w^{(\ell)} \beta^{\ell})_i| 
> 
\epsilon_i / (10\log n).
\label{eq:heavyhitter}
\end{align}
Here the first equality uses the fact that for $i\notin F^k$ the value of $d^{(\ell)}_i,d'^{(\ell)}_i,w^{(\ell)}_i$ stayed the same for $\ell=t-2^k+1,...,t$.

As $J_k = I_k \cup F^k$, the set $J_k$ contains all $i\in[n]$ that satisfy \eqref{eq:heavyhitter}.
The algorithms sets $\os_i \leftarrow s^{(t)}_i$ for all $i\in J_k$, so for these indices $i$ the approximation guarantee $|\os_i - s^{(t)}_i| < \epsilon_i$ holds.

Now consider $i\notin J_k$ and let $t'<t$ be the last time we set $\os_i \leftarrow s^{(t')}_i$.
Then there is a sequence of length at most $\log(t'-t)$ many $t_1<t_2<...<t_p$
where $t_j - t_{j-1}$ is a power of two, and $t_1 = t'$, $t_p = t$.
More accurately, these $t_j$ are the time steps during which we previously had that
$|s^{(t_j)}_i-s^{(t_{j-1})}_i| \le \epsilon_i$
(as otherwise we would have set $\os_i \leftarrow s^{(t_j)}_i$).
Thus by triangle inequality we have 
$$|s^{(t)}_i - \os_i| = |s^{(t)} - s^{(t')}_i| < \log(t-t') \epsilon_i / (10 \log m).$$
Note that here we used that $\epsilon_i$ stayed the same as otherwise $i$ was added to $F^1$ and $\os_i$ would have been updated.
For $t < \sqrt{m}$, this implies $|\os_i - s^{(t)}_i| < \epsilon_i$.

\end{proof}

\begin{lemma}\label{lem:vector:entrybound}
Consider an execution of \Cref{line:vector:k} in \Cref{alg:vector:maintenance}
for some $k$.
Let $I_k \subset [n]$ be the sets after line \ref{line:vector:heavyhitter} and \ref{line:vector:gradient}.
Then
$
|I_k|
=
\tO(
2^{2k}
).
$
\end{lemma}

\begin{proof}
Note that $I_k \cap F^k = \emptyset$ and for $i\notin F^k$, the value of $d^{(\ell)}_i,d'^{(\ell)}_i,\epsilon^{(\ell)}_i$ stayed the same for $\ell=t-2^k+1,...,t$, so we can write
\begin{align*}
|I_k| =&~ \sum_{i\notin F^k} \mathbf{1}_{i \in I_k} \\
\le&~
\sum_{i\notin F^k} 
\left(
(\mG\mA \sum_{\ell=t-2^k+1}^t h^{(\ell)})_i^2
+
(\mG'\mA \sum_{\ell=t-2^k+1}^t h'^{(\ell)})_i^2
+
(\epsilon_i^{-1}w_i\cdot\sum_{\ell=t-2^k}^t \beta^{(\ell)})^2\right) / (30 \log m)
\\
=&~
\tO\left(
\sum_{i\not\in F^k} 
\left(
\frac{(\sum_{\ell=t-2^k+1}^t \mD^{(\ell)}\mA h^{(\ell)})_i^2}{(\epsilon_i^{(\ell)})^2}
+
\frac{(\sum_{\ell=t-2^k+1}^t \mD'^{(\ell)}\mA h'^{(\ell)})_i^2}{(\epsilon_i^{(\ell)})^2}
+
\frac{(\sum_{\ell=t-2^k+1}^t w^{(\ell)}\beta^{(\ell)})_i^2}{(\epsilon_i^{(\ell)})^2}
\right)
\right)
\\
\le&~
\tO\left(
2^k
\sum_{\ell=t-2^k+1}^{t} (
\|(\epsilon^{(\ell)})^{-1}\mD\mA h^{(\ell)}\|^2_2
+
\|(\epsilon^{(\ell)})^{-1}\mD'\mA h'^{(\ell)}\|^2_2
+
\|(\epsilon^{(\ell)})^{-1}w^{(\ell)}\beta^{(\ell)}\|^2_2
)
\right)
\\
\le&~
\tO(2^{2k}).
\end{align*}
\end{proof}

\begin{lemma}\label{lem:vector:change_bound}
Every $2^k$ iterations, at most $\tO(2^{2k}+C)$ entries change in $\os$ where $C$ is the number of calls to \textsc{SetAccuracy} and \textsc{Update} over the past $2^k-1$ iterations.

\end{lemma}

\begin{proof}
After $2^k$ iterations we change the entries $\os_i$ for $i \in F^k$ and those $i$ detected as large in \Cref{line:vector:heavyhitter,line:vector:gradient} (i.e. $i \in I_k$).
For any such $k$, the number of entries in $I_k$
is bounded by $\tO(2^{2k})$ by \Cref{lem:vector:entrybound}.
The indices in $F^k$ are those for which we called \textsc{SetAccuracy} or \textsc{Update} between the past $2^k-1$ calls to \textsc{Add}.
Hence we obtain the bound $\tO(2^{2k}+C)$ on the number of changed entries in $\os$.

\end{proof}

\begin{lemma}
    The amortized cost of a call to \textsc{SetAccuracy} and \textsc{Update} is $\hO(1)$.
\end{lemma}

\begin{proof}
    A call to \textsc{Update} of \Cref{lem:vector:implicit:sumofproduct,lem:vector:implicit:sumofvector} has $\hO(1)$ amortized cost each.
    Adding the index to $F^k$ later causes a call to \textsc{Scale} of \Cref{lem:vector:heavyhitter} in \Cref{line:Vector:update_g} because we must update the vector $g$.
    and a call to \textsc{Query} of \Cref{lem:vector:implicit:sumofproduct,lem:vector:implicit:sumofvector} in \Cref{line:vector:update_os}.
    As these methods also have cost $\hO(1)$ and there are only $O(\log n)$ different $F^k$, in total we have $\hO(1)$ amortized cost.
\end{proof}

\begin{lemma}
For the first $\sqrt{m}$ calls to \textsc{Add}, the amortized cost per call is
$$\tO\left(
m^{1/2+o(1)} + n \log W
\right)$$
where $W$ is a bound on the ratios of largest to smallest entry in $d/\epsilon$.
\end{lemma}

\begin{proof}
A call to \textsc{Add} also causes a call to \textsc{Add} of \Cref{lem:vector:implicit:sumofproduct,lem:vector:implicit:sumofvector} which has $O(1)$ complexity.

Now consider a loop of \Cref{line:vector:k}  for some $k$.
Detecting indices $i$ in \Cref{line:vector:heavyhitter} takes $$\tO(\|\mG\mA\sum_{\ell=2^k+1}^th^{(\ell)}\|_2^2)$$
by \Cref{lem:vector:heavyhitter}.
This can be bounded by
$$
\tO(2^{2k} (\sum_{\ell=t-2^k+1}^t \|(\epsilon^{(\ell)})^{-1}\mD^{(\ell)}h^{(\ell)})\|_2^2 + n \log W)
$$
via Cauchy-Schwarz and the fact that $\epsilon^{(\ell)}_i,d^{(\ell)}_i$ stayed the same for $\ell=t-2^k+1,...,t$ and $i\notin F^k$.
As we consider a $k$ at most once every $2^k$ iterations and we are promised that the norm above is bounded by a constant,
and further $0\le k\le \log \sqrt{m}$
this leads to amortized cost of $\tO(\sqrt{m} + n \log W)$ per call to \textsc{Add}.

Note that this assumes we previously called \textsc{Scale} of \Cref{lem:vector:heavyhitter} to insert the right diagonal matrix $\mG$.
This cost was already charged to \textsc{SetAccuracy} and \textsc{Update}.

The cost of \Cref{line:vector:gradient} is just the number of returned indices.
This is because we can implement the task of finding $i\notin F^k$ with $|w_i\sum_{\ell=t-2^k+1}^t \beta^{(\ell)}$ via a priority queue.
Simply maintain an order of the $w_i$ for all $i\notin F^k$ throughout all calls to \textsc{Update} and \textsc{SetAccuracy}.
This adds only $\hO(1)$ amortized cost to \textsc{Update} and \textsc{SetAccuracy}.

Updating $\os_i$ in \Cref{line:vector:update_os} takes $O(1)$ per index and there are at most $\tO(2^{2k})$ such indices every $2^k$ iterations by \Cref{lem:vector:change_bound}. (Where the cost incurred by previous calls to \textsc{Update} and \textsc{SetAccuracy} is again charged as amortized cost to those functions.)
For $2^k\le\sqrt{m}$ this again implies $\tO(m^{1/2+o(1)}+n\log W)$ amortized cost per call to \textsc{Add}.
\end{proof}

\end{document}